\documentclass[11pt,onecolumn]{IEEEtran}

\usepackage{geometry} % to change the page dimensions
\usepackage{graphicx} % support the \includegraphics command and options
\usepackage{color}

\definecolor{Red}{rgb}{1,0,0}
\definecolor{Blue}{rgb}{0,0,1}
\definecolor{Olive}{rgb}{0.41,0.55,0.13}
\definecolor{Green}{rgb}{0,1,0}
\definecolor{MGreen}{rgb}{0,0.8,0}
\definecolor{DGreen}{rgb}{0,0.55,0}
\definecolor{Yellow}{rgb}{1,1,0}
\definecolor{Cyan}{rgb}{0,1,1}
\definecolor{Magenta}{rgb}{1,0,1}
\definecolor{Orange}{rgb}{1,.5,0}
\definecolor{Violet}{rgb}{.5,0,.5}
\definecolor{Purple}{rgb}{.75,0,.25}
\definecolor{Brown}{rgb}{.75,.5,.25}
\definecolor{Grey}{rgb}{.5,.5,.5}

\usepackage{booktabs} % for much better looking tables
\usepackage{array} % for better arrays (eg matrices) in maths
\usepackage{paralist} % very flexible & customisable lists (eg. enumerate/itemize, etc.)
\usepackage{verbatim} % adds environment for commenting out blocks of text & for better verbatim
\usepackage{subfigure} % make it possible to include more than one captioned figure/table in a single float
\usepackage{amsmath,amssymb,amsthm}
\theoremstyle{plain}

\newtheorem{theorem}{Theorem} 
 
\newtheorem{corollary}{Corollary}
\newtheorem{claim}{Claim}
\newtheorem{lemma}{Lemma}

\theoremstyle{remark}

\newtheorem{remark}{Remark}

\theoremstyle{definition}
\newtheorem{definition}{Definition}

\usepackage{tikz}
\usetikzlibrary{positioning}

\usepackage{hyperref}
\hypersetup{ pdfpagemode=FullScreen,  colorlinks=true}

\newcommand{\Ac}{\mathcal{A}}
\newcommand{\Bc}{\mathcal{B}}
\newcommand{\Cc}{\mathcal{C}}

\newcommand{\Gc}{\mathcal{G}}

\newcommand{\Pc}{\mathcal{P}}

\newcommand{\Uc}{\mathcal{U}}
\newcommand{\Vc}{\mathcal{V}}
\newcommand{\Wc}{\mathcal{W}}
\newcommand{\Xc}{\mathcal{X}}
\newcommand{\Yc}{\mathcal{Y}}
\newcommand{\Zc}{\mathcal{Z}}

\newcommand{\Ut}{{\tilde{U}}}
\newcommand{\Vt}{{\tilde{V}}}
\newcommand{\Wt}{{\tilde{W}}}

\newcommand{\ut}{{\tilde{u}}}
\newcommand{\vt}{{\tilde{v}}}
\newcommand{\wt}{{\tilde{w}}}

\def\a{\alpha}

\def\e{\epsilon}

\def\la{\lambda}

\let\P\relax
\DeclareMathOperator\P{P}
\def\textiid{i.i.d.\@\xspace}
\newcommand\iid{\ifmmode\text{ i.i.d. } \else \textiid \fi}

\newcommand{\qmf}{\mathfrak{q}}
\newcommand{\mar}{\to}
\newcommand{\dash}{\mbox{-}}

\newcommand{\beqs}{\begin{equation*}}
\newcommand{\eeqs}{\end{equation*}}
\newcommand{\beq}{\begin{equation}}
\newcommand{\eeq}{\end{equation}}

\title{On Marton's inner bound and its optimality for classes of product broadcast channels}

\author{Yanlin~Geng,~\IEEEmembership{Member,~IEEE,}
        Amin~Gohari,~\IEEEmembership{Member,~IEEE,}
        Chandra~Nair,~\IEEEmembership{Member,~IEEE,}
        and~Yuanming~Yu % <-this % stops a space
\thanks{The work of A.~Gohari was partially supported by the following grant from Iran National Science Foundation, No. 89003743.}%
\thanks{The work of C.~Nair was partially supported by the following grants from the University Grants Committee of the Hong Kong Special Administrative Region, China: a) (Project No. AoE/E-02/08), b) GRF Project 415810. He also acknowledges the support from the Institute of Theoretical Computer Science and Communications (ITCSC) at The Chinese University of Hong Kong.}%
\thanks{This paper contains the proofs of the results that were in part presented at ITA, UCSD and at ISIT, 2011.}
}
\date{}

\newcommand{\p}{{\rm P}}

\begin{document}
\maketitle
\begin{abstract}
Marton's inner bound is the tightest known inner bound on the capacity region of the broadcast channel. It is not known, however, if this bound is tight in general. One approach to settle this key open problem in network information theory is to investigate the multi-letter extension of Marton's bound, which is known to be tight in general. This approach has become feasible only recently through the development of a new method for bounding cardinalities of auxiliary random variables by Gohari and Anantharam.  This paper undertakes this long overdue approach to establish several new results, including $(i)$ establishing the optimality of Marton\rq{}s bound for new classes of product broadcast channels, (ii) showing that the best known outer bound by Nair and El Gamal is not tight in general, and (iii) finding sufficient conditions for a global maximizer of Marton's bound that imply that the 2-letter extension does not increase the achievable rate. Motivated by the new capacity results, we establish a new outer bound on the capacity region of product broadcast channels in general. 
\end{abstract}

\section{Introduction}
Consider the broadcast channel $\qmf(y,z|x)$ with private messages depicted in Figure \ref{fig:bc}. The sender $X$ wishes to communicate a message $M_1$ at rate $R_1$ to receiver $Y_1$ and a message $M_2$ at rate $R_2$ to another receiver $Y_2$. What is the capacity region, that is, the closure of the set of achievable rate pairs $(R_1,R_2)$?

This question is one of the key open problems in network information theory. Since the introduction of this problem in the groundbreaking paper by Cover~\cite{cov72},  several inner and outer bounds on the capacity region of this channel have been developed and shown to be tight in some special cases; see Chapters 5, 8, and 9 of~\cite{elk12} for a detailed discussion of previous works.

Marton's inner bound~\cite{mar79} and the UV outer bound~\cite{nae07}(also sometimes referred to as the Nair--El Gamal outer bound) are the tightest known bounds on the capacity region of the broadcast channel. These bounds have been shown to coincide for all classes of broadcast channels with known capacity regions. Recently it has been shown \cite{naw08,goa09,jon09} that there are channels for which these inner and outer bounds do not coincide. Therefore, clearly at least one of them is strictly sub-optimal.

In this paper we show that the UV outer bound is strictly suboptimal by establishing the capacity region for a new class of broadcast channels and showing that this capacity region coincides with Marton's inner bound but not with the UV outer bound. This result is only one consequence of exploring an approach to establish the optimality (or lack thereof)  of Marton's region by investigating its multi-letter extension. This approach, although conceptually simple, has only become interesting recently. This is due to the fact that cardinality bounds on the auxiliary random variables in Marton\rq{}s inner bound were recently established in \cite{goa09}; and only since then did Marton\rq{}s inner bound become computable and hence amenable to numerical simulations for test channels.

\subsection{Preliminaries}
\label{sse:prelim}

%%% [Define problem (as per NIT book), state inner and outer bounds, and describe results with examples and brief proof ideas.]

\begin{figure}[h]\centering
\begin{tikzpicture}[scale=1.3]
\node at (0,0) {$(M_1,M_2)$};
\draw [->,thick] (.8,0) -- (1.6,0);
\draw (1.6,-.5) rectangle +(1.6,1); \node at (2.4,0) {Encoder};
\draw [->,thick] (3.2,0)--(4,0); \node at (3.6,.3) {$X^n$};
\draw (4,-1.5) rectangle +(1.6,3); \node at (4.8,0) {$\qmf(y,z|x)$};
\draw [->,thick] (5.6,1) --(6.4,1); \node at (6,1.3) {$Y^n$};
\draw [->,thick] (5.6,-1) --(6.4,-1); \node at (6.1,-.7) {$Z^n$};
\draw (6.4,.5) rectangle +(2,1); \node at (7.4,1) {Decoder 1};
\draw (6.4,-1.5) rectangle +(2,1); \node at (7.4,-1) {Decoder 2};
\draw [->,thick] (8.4,1) --(9.2,1); \node at (9.6,1) {$\hat M_1$};
\draw [->,thick] (8.4,-1) --(9.2,-1); \node at (9.6,-1) {$\hat M_2$};
\end{tikzpicture}
\caption{A broadcast channel}
\label{fig:bc}
\end{figure}

In the broadcast channel setting  a sender $X$, who has messages $M_1, M_2$, wishes to communicate message $M_1$ to receiver $Y$ and $M_2$ to receiver $Z$ over a noisy discrete memoryless broadcast channel $\qmf(y,z|x)$.
A set of rate pairs $(R_1, R_2)$ is said to be achievable for  this broadcast channel, $\qmf(y,z|x)$,  if there is a sequence of codebooks, each consisting of:
\begin{itemize}
\item an encoder at the sender that maps the message pair $(M_1, M_2)$ into a sequence $X^n,$
\item a decoder at receiver $Y$ that maps the received sequence $Y^n$ into an estimate $\hat{M}_1$ of its intended message $M_1$, and
\item a decoder at receiver $Z$ that maps the received sequence $Z^n$ into an estimate $\hat{M}_2$ of its intended message $M_2$
\end{itemize}
such that $\P( \hat{M}_1 \neq M_1),  \P( \hat{M}_2 \neq M_2) \to 0$ as $n \to \infty$, when the messages $M_1, M_2$ are uniformly distributed in $[1:2^{nR_1}] \times [1:2^{nR_2}]$. The capacity region is the closure of the set of all achievable rate pairs. An evaluable characterization of this capacity region is a well known open problem.

An inner bound to the capacity region refers to a set of rate pairs for which there is a strategy to achieve it.  The best known inner bound to the capacity region of the two receiver broadcast channel is due to Marton \cite{mar79}. It is not known if Marton's inner bound is optimal or not. Marton's inner bound for a general two-receiver discrete-memoryless broadcast channel with private messages is the following:

{\em Inner bound}: (Marton \cite{mar79})
The union of rate pairs $(R_1, R_2)$ satisfying the inequalities
\begin{align}
R_1 &\leq I(U,W;Y) \nonumber \\
R_2 & \leq I(V,W;Z) \label{eq:mib}\\
R_1 + R_2 &\leq  \min\{I(W;Y), I(W;Z)\} + I(U;Y|W) + I(V;Z|W) - I(U;V|W) \nonumber
\end{align}
over all $(U,V,W,X): (U,V,W) \mar X \mar (Y,Z)$ forms a Markov chain constitutes an inner bound to the capacity region. Further to compute this region it suffices \cite{goa09} to consider $|\Uc|, |\Vc| \leq |\Xc|, |\Wc| \leq |\Xc| + 4.$

One of the main results of this paper is computing the capacity region for new classes of product broadcast channels, and deducing that the best  outer bound previously known is strictly sub-optimal. The best  outer bound\footnote{Though there have been several proposed outer bounds since \cite{nae07}, it was shown in \cite{nai11} that they reduced to the one in \cite{nae07} for the private messages case.} for a general two-receiver discrete-memoryless broadcast channel with private messages is the following:

{\em Outer bound}: (UV outer bound \cite{nae07})
The union of rate pairs $(R_1, R_2)$ satisfying the inequalities
\begin{align*}
R_1 &\leq I(U;Y) \\
R_2 & \leq I(V;Z) \\
R_1 + R_2 &\leq  I(U;Y) + I(X;Z|U) \\
R_1 + R_2 &\leq  I(V;Z) + I(X;Y|V) \\
\end{align*}
over all $(U,V,X): (U,V) \mar X \mar (Y,Z)$ forms a Markov chain constitutes an outer bound to the capacity region. Further to compute this region it suffices to consider $|\Uc|, |\Vc| \leq |\Xc|+1.$

For the case with  common and private messages  requirement  one has the following outer bound for the capacity region.

{\em Outer bound}: (UVW outer bound \cite{nai11})
The union of rate triples $(R_0,R_1,R_2)$ satisfying the inequalities
\begin{align*}
 R_0 &\leq \min \{I(W;Y), I(W;Z)\}\\
 R_0+R_1 &\leq I(U;Y|W) + \min \{I(W;Y), I(W;Z)\}\\\
 R_0 + R_2 &\leq I(V;Z|W)+ \min \{I(W;Y), I(W;Z)\}\\ 
 R_0+R_1+ R_2 &\leq \min \{I(W;Y), I(W;Z)\} + I(X;Y|V,W)+ I(V;Z|W)\\
 R_0+R_1+ R_2 &\leq \min \{I(W;Y), I(W;Z)\} + I(U;Y|W)+ I(X;Z|U,W)
\end{align*}
over all $(U,V,W,X)$ such that $(U,V,W) \mar X \mar (Y,Z)$
forms a Markov chain constitutes an outer bound to the capacity region. Further  it suffices to consider $|\Wc| \leq |\Xc| + 5, |\Uc| \leq |\Xc| + 1, |\Wc| \leq |\Xc| + 1$.

It is also established in \cite{nai11} that when $R_0=0,$ the UVW outer bound reduces to the UV outer bound.

\subsubsection{Definitions of some classes of  broadcast channels} \label{sec:defn}
\begin{definition}
A broadcast channel $\qmf(y,z|x)$ is said to be a {\em product broadcast channel} if $\Xc = (\Xc_1, \Xc_2), \Yc = (\Yc_1, \Yc_2), \Zc=(\Zc_1, \Zc_2)$ and $\qmf(y_1, y_2, z_1, z_2|x_1,x_2) = \qmf_1(y_1,z_1|x_1)\qmf_2(y_2,z_2|x_2)$. Here we denote $\qmf = \qmf_1 \times \qmf_2$.
\end{definition}

\begin{definition}
A product broadcast channel $\qmf=\qmf_1 \times \qmf_2$ is said to be {\em reversely semi-deterministic} if  the channel to one of the receivers in the first component is deterministic, and the channel to the other receiver in the second component is deterministic. That is either both $\qmf_1(y_1|x_1), \qmf_2(z_2|x_2) \in \{0,1\}$ or both $\qmf_1(z_1|x_1), \qmf_2(y_2|x_2) \in \{0,1\}$.
\end{definition}

\begin{definition}
A product broadcast channel $\qmf=\qmf_1 \times \qmf_2$ is said to be {\em reversely more capable} if  one of the following two holds:
\begin{itemize}
\item $I(X_1;Y_1) \geq I(X_1;Z_1), ~ \forall p(x_1)$, and $ I(X_2;Z_2) \geq I(X_2;Y_2), ~ \forall p(x_2),$
\item $I(X_1;Z_1) \geq I(X_1;Y_1), ~ \forall p(x_1)$, and $ I(X_2;Y_2) \geq I(X_2;Z_2), ~ \forall p(x_2).$
\end{itemize}
\end{definition}

%%% [Following sections and appendices provide detailed proofs]

\subsection{Organization and summary of results}
The rest of the paper is organized as follows. In Section \ref{sec:UVloose} we show that the UV outer bound is not tight (Claim \ref{cl:uvsub}). To do so, we need to introduce a quantity called the $\la$-sum-rate and use some of its properties to compare the inner and outer bounds. In Section \ref{subsec:lambdasumrate} we will establish some additional properties of $\la$-sum-rate some of which will be used critically for explicit evaluations of bounds. In Section \ref{sec:capacityregions} we establish a new outer bound (Claim \ref{cl:obp}) for product broadcast channels and use it to determine the capacity region of some new classes (Theorems \ref{Thm2} and \ref{th:mcp}). This outer bound is strictly better than the UV outer bound as it is optimal for the example where the UV outer bound is loose. In  (Lemma \ref{le:mnf}) we show that Marton\rq{}s region for a product of two non-identical broadcast channels can be strictly larger than the (Minkowski) sum of the individual regions.
Section \ref{sec:randomziedtimedivision} deals with a particular coding strategy, randomized time-division, that is equivalent to Marton's inner bound (albeit much simpler) for binary input broadcast channels. If the generalization of the strategy described in Section \ref{sec:randomziedtimedivision} is indeed equivalent to Marton's inner bound for products of binary input broadcast channels, then one can deduce the optimality of Marton's coding scheme for binary input broadcast channels from a result (Theorem \ref{th:rtdf}) we establish in this section. Essentially, the results in this section show a reduction from an $n$-letter characterization to a single-letter characterization within a  randomized time-division strategy. Finally,  some of the technical arguments as well as other lengthy but routine arguments are relegated to the Appendices.

\section{The UV outer bound is not tight}
\label{sec:UVloose}

The flow of this section is as follows: we first introduce {\it $\la$-sum-rate}, a quantity that helps in the computation of Marton\rq{}s inner bound. To explicitly compute the sum rate for product channels we introduce the notion of {\it factorization of  $\la$-sum-rate}. Using this factorization idea, we show that Marton\rq{}s sum rate is optimal for the product of reversely semi-deterministic channels. Having computed the optimal sum rate,  the UV outer bound is shown to be strictly suboptimal (via a specifically constructed example) over this class of broadcast channels.

\subsection{Definitions and preliminary results}
\label{sse:prelimd}
Given a broadcast channel $\qmf(y,z|x)$ we define the following quantities for $\la \in [0,1]$ and for auxiliary random variables $(U,V,W)$ that satisfy the Markov chain $(U,V,W) \mar X \mar (Y,Z)$:
{\small \begin{align}
 \la\dash SR_M(\qmf,p(u,v,w,x)) & :=  \la I(W;Y) + (1-\la) I(W;Z)   + I(U;Y|W) + I(V;Z|W) - I(U;V|W)
\label{eq:lmib00lm}
\end{align}}
 \begin{align}
 \la\dash SR_M(\qmf,p(x)) & := \max_{\substack{p(u,v,w|x):\\(U,V,W)\to X\to (Y,Z)}} \la I(W;Y) + (1-\la) I(W;Z)   + I(U;Y|W) \nonumber \\
 & \qquad \qquad \qquad  \qquad \qquad \qquad+ I(V;Z|W) - I(U;V|W)
\label{eq:lmib1lm}
\end{align}
 \begin{align}
 \la\dash SR_M(\qmf) & := \max_{\substack{p(u,v,w,x):\\(U,V,W)\to X\to (Y,Z)}} \la I(W;Y) + (1-\la) I(W;Z)   + I(U;Y|W) \nonumber\\
 & \qquad \qquad \qquad  \qquad \qquad \qquad+ I(V;Z|W) - I(U;V|W)
\label{eq:lmib1l}
\end{align}
Note the following relations: 
$$ \la\dash SR_M(\qmf,p(x)) = \max_{\substack{p(u,v,w|x):\\(U,V,W)\to X\to (Y,Z)}}\la\dash SR_M(\qmf,p(u,v,w,x)),$$
and $$ \la\dash SR_M(\qmf) = \max_{p(x)} \la\dash SR_M(\qmf,p(x)).$$ Further one can verify  that $ \la\dash SR_M(\qmf,p(x))$ is concave in $p(x)$ for a fixed $\la$.

Note that the maximum sum rate yielded by Marton\rq{}s inner bound in \eqref{eq:mib} is given by
$$ SR_M(\qmf) := \max_{\substack{p(u,v,w,x):\\(U,V,W)\to X\to (Y,Z)}} \min\{I(W;Y), I(W;Z)\}    + I(U;Y|W) + I(V;Z|W) - I(U;V|W). $$
Hence $SR_M(\qmf) = \max_{p(u,v,w,x)} \min_{\la \in [0,1]} \la\dash SR_M(\qmf,p(u,v,w,x))$.

The following lemma allows us to shift the discussion from Marton\rq{}s sum rate to $\la$-sum-rate, and then return to Marton\rq{}s sum rate at a later point to complete our arguments.

\begin{lemma} The following min-max theorem holds:
\label{le:mm}
\begin{align*}
\max_{p(u,v,w,x)} \min_{\la \in [0,1]} \la\dash SR_M(\qmf,p(u,v,w,x)) &  = \max_{p(x)} \min_{{ \la \in [0,1]}} \max_{p(u,v,w|x)} \la\dash SR_M(\qmf,p(u,v,w,x)) \\
&  = \min_{{\la \in [0,1]}} \max_{p(u,v,w,x)}  \la\dash SR_M(\qmf,p(u,v,w,x)).
\end{align*}
This implies that the sum rate of Marton's inner bound can be calculated using any of the three above expressions.
\end{lemma}
\begin{proof} The proof is presented in  Appendix \ref{Apendix0} and can be considered as an application of a min-max theorem of Terkelsen\cite{ter72}. The fact that \begin{align*}SR_M(\qmf) = \max_{p(u,v,w,x)} \min_{\la \in [0,1]} \la\dash SR_M(\qmf,p(u,v,w,x)) = \min_{{\la \in [0,1]}} \max_{p(u,v,w,x)}  \la\dash SR_M(\qmf,p(u,v,w,x))\end{align*} was also established in section 3.1.1 of \cite{gea10}. However Corollary \ref{coro:mm} established in the Appendix, which is the crux of the current proof, can also be used in other instances where a max-min occurs, such as compound channels. 
\end{proof}

\begin{definition}
For a given product channel $\qmf_1(y_1,z_1|x_1) \times \qmf_2(y_2,z_2|x_2)$ we say that the 
{\em $\la$-sum-rate factorizes} if for all $p(x_1,x_2)$ we have
\begin{equation}
\la \dash SR_M(\qmf_1 \times \qmf_2, p_{X_1,X_2}(x_1,x_2))\leq \la \dash SR_M(\qmf_1,p_{X_1}(x_1)) + \la \dash SR_M(\qmf_2,p_{X_2}(x_2)).
\label{eq:lsumf}
\end{equation}
\end{definition}

\subsubsection*{Sufficient conditions for factorization of $\la$-sum-rate}

In this section, we derive sufficient conditions under which \eqref{eq:lsumf} holds.
The following claim is key to the arguments  in this section.
\begin{claim}
\label{cl:one}
Let $U_1=U_2=U, V_1=V_2=V, W_1=(W,Z_2), W_2=(W,Y_1)$. Then the following holds:
 \begin{align*} &\la \dash SR_M(\qmf_1 \times \qmf_2, p(u,v,w,x_1,x_2)) \\
&\quad =  \la \dash SR_M(\qmf_1, p(u_1,v_1,w_1,x_1)) + \la \dash SR_M(\qmf_2, p(u_2,v_2,w_2,x_2))  + I(U;V|W,Y_1,Z_2)\\
& \qquad - \la I(Y_1;Y_2) - (1-\la) I(Z_1;Z_2) - I(Y_1;Z_2|U,V,W).
\end{align*}
\end{claim}
\begin{proof}
\begin{align*}
& \la \dash SR_M(\qmf_1 \times \qmf_2, p(u,v,w,x_1,x_2))\\
&\quad = \la I(W;Y_1,Y_2) + (1-\la) I(W;Z_1, Z_2) + I(U;Y_1, Y_2|W) + I(V;Z_1, Z_2|W) - I(U;V|W) \\
& \quad = \la I(W, Z_2;Y_1) + (1-\la) I(W, Z_2;Z_1) + I(U;Y_1|W,Z_2) + I(V;Z_1|W, Z_2) - I(U;V|W,Z_2) \\
& \quad \quad +  \la I(W,Y_1; Y_2) + (1-\la) I(W,Y_1; Z_2) + I(U;Y_2|W, Y_1) + I(V; Z_2|W, Y_1) - I(U;V|W, Y_1) \\
& \quad \quad + I(U;V|W,Y_1,Z_2) - \la I(Y_1;Y_2) - (1-\la) I(Z_1;Z_2) - I(Y_1;Z_2|U,V,W). \qedhere
\end{align*} 
\end{proof}

Thus the excess term one needs to cancel (using a different choice of $(U_1,V_1,W_1)$ or $(U_2,V_2,W_2)$ or both) to ensure factorization, is at most $ I(U;V|W,Y_1,Z_2)$.

Also observe that one can get a similar identity by interchanging $Y_1 \leftrightarrow Z_1$ and $Z_2 \leftrightarrow Y_2$. Here $W_1=(W,Y_2)$ and $W_2=(W,Z_1)$. This will yield the term $ I(U;V|W,Y_2,Z_1)$ instead of $ I(U;V|W,Y_1,Z_2)$.

\begin{theorem}
\label{th:fact}
The $\la$-sum-rate factorizes (as in \eqref{eq:lsumf}) if either of the conditions below hold:
\begin{enumerate}
\item Any one of the four channels $X_1 \to Y_1; X_1 \to Z_1; X_2 \to Y_2$ or $X_2 \to Z_2$ is deterministic.
\item In either  of the two components, one channel is more capable than the other.
\end{enumerate}
\label{th:1suff}
\end{theorem}
\begin{proof}
Assume the first condition holds. In particular let $X_2 \to Z_2$ be deterministic. Then we will show that
\begin{align*}
& \la \dash SR_M(\qmf_1 \times \qmf_2, p(u,v,w,x_1,x_2)) \leq \la \dash SR_M(\qmf_1, p(u_1,v_1,w_1,x_1)) + \la \dash SR_M(\qmf_2, p(u_2,v_2,w_2,x_2))
\end{align*}
where $U_1=U_2=U, V_1=V, V_2=Z_2, W_1=(W,Z_2), W_2=(W,Y_1)$. To show this, from Claim \ref{cl:one} it suffices to show that
$$ \la \dash SR_M(\qmf_2, p(u,v,w_2,x_2)) + I(U;V|W,Y_1,Z_2) \leq \la \dash SR_M(\qmf_2, p(u,z_2,w_2,x_2)), $$
where $w_2=(w,y_1)$.
Observe that
\begin{align*}
& \la \dash SR_M(\qmf_2, p(u,v,w_2,x_2)) + I(U;V|W,Y_1,Z_2)  \\
& \quad = \la I(W,Y_1; Y_2) + (1-\la) I(W,Y_1; Z_2) + I(U;Y_2|W, Y_1) \\
& \quad \quad + I(V; Z_2|W, Y_1) - I(U;V|W, Y_1) + I(U;V|W,Y_1,Z_2)  \\
& \quad = \la I(W,Y_1; Y_2) + (1-\la) I(W,Y_1; Z_2) + I(U;Y_2|W, Y_1) + I(V; Z_2|U,W, Y_1) \\
& \quad \leq \la I(W,Y_1; Y_2) + (1-\la) I(W,Y_1; Z_2) + I(U;Y_2|W, Y_1) + H(Z_2|U,W, Y_1) \\
& \quad = \la I(W,Y_1; Y_2) + (1-\la) I(W,Y_1; Z_2) + I(U;Y_2|W, Y_1) + I(Z_2;Z_2|W, Y_1) - I(Z_2;U|W,Y_1) \\
&\quad = \la \dash SR_M(\qmf_2, p(u,z_2,w_2,x_2)),
\end{align*}
where we identify $w_2=(w,y_1)$.
Similar reasoning can deal with the case where $X_1 \to Y_1$ is a deterministic channel.

Note that if $X_2 \to Y_2$ is deterministic, then one must start with the interchanged $W_1, W_2$, i.e. $W_1=(W,Y_2), W_2=(W,Z_1)$, and similarly show that
$$ \la \dash SR_M(\qmf_2, p(u,v,w_2,x_2)) + I(U;V|W,Z_1,Y_2) \leq \la \dash SR_M(\qmf_2, p(y_2,v,w_2,x_2)), $$ where $w_2=(w,z_1)$. Finally, the case when $X_1 \to Z_1$ is deterministic can be dealt with similarly.

Proceeding to the second condition, let us assume that the channel $X_2 \to Y_2$ is more capable than the channel $X_2 \to Z_2$, i.e. for all $p(x_2), I(X_2;Y_2) \geq I(X_2;Z_2).$ Let $W_2=(W,Y_1)$. Then observe that
\begin{align*}
& \la \dash SR_M(\qmf_2, p(u,v,w_2,x_2)) + I(U;V|W,Y_1,Z_2)  \\
& \quad = \la I(W,Y_1; Y_2) + (1-\la) I(W,Y_1; Z_2) + I(U;Y_2|W, Y_1) \\
& \quad \quad + I(V; Z_2|W, Y_1) - I(U;V|W, Y_1) + I(U;V|W,Y_1,Z_2)  \\
& \quad = \la I(W,Y_1; Y_2) + (1-\la) I(W,Y_1; Z_2) + I(U;Y_2|W, Y_1) + I(V; Z_2|U,W, Y_1) \\
& \quad \leq \la I(W,Y_1; Y_2) + (1-\la) I(W,Y_1; Z_2) + I(U;Y_2|W, Y_1) + I(X_2; Z_2|U,W, Y_1) \\
& \quad \leq \la I(W,Y_1; Y_2) + (1-\la) I(W,Y_1; Z_2) + I(U;Y_2|W, Y_1) + I(X_2; Y_2|U,W, Y_1) \\
& \quad = \la I(W,Y_1; Y_2) + (1-\la) I(W,Y_1; Z_2) + I(X_2;Y_2|W, Y_1)\\
& \quad = \la \dash SR_M(\qmf_2, p(x_2,\emptyset,w_2,x_2)).
\end{align*}
Thus from Claim \ref{cl:one} we have the factorization of $\la \dash SR_M(\qmf_1 \times \qmf_2)$.

Similar reasoning works for the other three cases. Again observe that when $Z_2$ is more capable than $Y_2$ or $Y_1$ is more capable than $Z_1$, one should start with start with the interchanged $W_1, W_2$, i.e. $W_1=(W,Y_2), W_2=(W,Z_1)$. This completes the proof of the lemma.
\end{proof}

\subsection{Optimal sum rate for product of reversely semi-deterministic channels}
\begin{claim}
\label{cl:reverselysemideterministic}
Marton's sum rate is optimal for  the product of reversely semi-deterministic channels. Moreover the sum rate of such a product channel $\qmf_1\times \qmf_2$ is given by
\begin{align*}&\min_{\lambda\in[0,1]}\big(\la \dash SR_M(\qmf_1) + \la \dash SR_M(\qmf_2)\big).\end{align*}
\end{claim}
\begin{proof}
Take two semi-deterministic channels $\qmf_1(y_1,z_1|x_1)$ and $\qmf_2(y_2,z_2|x_2)$ where $Y_1$ is a deterministic function of $X_1$ and $Z_2$ is a deterministic function of $X_2$. 

Consider the $n$-letter $\la$-sum-rate of the product channel $\qmf_1\times \qmf_2$. Using Theorem \ref{th:1suff}  the $\la$-sum-rate of $n$-letter product channel factorizes into $\la$-sum-rate of two $n$-letter sub channels. Each term  again factorizes by repeated application of Theorem \ref{th:1suff}. More precisely,
\begin{align*}\la \dash SR_M(\qmf_1\otimes_{n} \times \qmf_2\otimes_{n})& = \la \dash SR_M(\qmf_1\otimes_{n})+ \la \dash SR_M(\qmf_2\otimes_{n})\\
&=n\cdot\la \dash SR_M(\qmf_1)+n\cdot\la \dash SR_M(\qmf_2).
\end{align*}
Marton's inner bound sum rate for the $n$-letter of the product channel $\qmf_1\times \qmf_2$ is equal to
\begin{align*}&\min_{\lambda\in[0,1]}\big(\la \dash SR_M(\qmf_1\otimes_{n} \times \qmf_2\otimes_{n})\big).\end{align*}
We can write the above expression as
\begin{align*}&n\cdot \min_{\lambda\in[0,1]}\big(\la \dash SR_M(\qmf_1)+\la \dash SR_M(\qmf_2)\big).\end{align*}
Therefore, the actual sum rate satisfies\footnote{We utilize the known fact that Marton's inner bound sum rate for the $n$-letter version of the channel approaches the optimal sum rate as $n$ goes to infinity.
}
\begin{align*}SR^*(\qmf_1 \times \qmf_2) &\leq \min_{\lambda\in[0,1]}\big(\la \dash SR_M(\qmf_1)+\la \dash SR_M(\qmf_2)\big).
\end{align*}

On the other hand, this sum rate is achievable since it is equal to the single letter Marton's inner bound for $\qmf_1\times \qmf_2$, i.e.
\begin{align*} SR^*(\qmf_1 \times \qmf_2)&=\min_{\lambda\in[0,1]}\la \dash SR_M(\qmf_1\times \qmf_2)=\min_{\lambda\in[0,1]}\big(\la \dash SR_M(\qmf_1)+\la \dash SR_M(\qmf_2)\big).\end{align*}
\end{proof}

\subsection{The UV outer bound is strictly suboptimal}

From  the UV outer bound the sum rate of a general broadcast channel can be bounded from above by
\begin{equation}
SR_{UV}(\qmf) = \max_{p(u,v,x)} \min \{ I(U;Y) + I(V;Z), I(U;Y) + I(X;Z|U), I(V;Z) + I(X;Y|V) \}.
\end{equation}
In this example we will demonstrate a product of reversely semi-deterministic channel, $\qmf = \qmf_1 \times \qmf_2$,  such that the optimal sum rate $SR^*(\qmf_1 \times \qmf_2)$ satisfies
$$ SR_M(\qmf_1 \times \qmf_2) = SR^*(\qmf_1 \times \qmf_2) < SR_{UV}(\qmf_1 \times \qmf_2). $$
This unequivocally shows that the UV outer bound is {\em strictly} suboptimal for the general broadcast channel.

\begin{remark}
Even if one were to consider the best outer bound with a common message requirement, the UVW outer bound \cite{nai11}, the fact that we are showing that the sum rate is strictly weak for the UV outer bound immediately implies the {\em strict} sub optimality of the UVW outer bound as well. To note this, observe that the projection of the UVW outer bound on the plane $R_0=0$ (which is shown in \cite{nai11} to be the UV outer bound)  is {\em strictly} suboptimal.
\end{remark}

\begin{claim}
\label{cl:uvsub}
Consider the reversely semi-deterministic channel in Figure \ref{fif:f1}. Assume that the transition probabilities are uniform across the possible outputs, i.e the red edges have a probability $\frac 13$ in the first component and the blue edges have a probability $\frac 13$ in the second component.  Then Marton's sum rate (the optimal sum rate) is given by $\frac 83=3-\frac13$ , while the UV sum rate is at least  $3 - \frac 1{15}$.
\begin{figure}
%\centering
%\includegraphics[width=40mm]{fig_3.pdf}
\begin{center}
\begin{minipage}{0.45\linewidth}
\centering
\resizebox{40mm}{!}{%
\begin{tikzpicture}[
    >=stealth,
    xnode/.style={draw, circle, minimum size=7mm},
    ynode/.style={draw, circle, minimum size=7mm},
    znode/.style={draw, circle, minimum size=7mm},
    every edge/.style={->, thick}
]

% X nodes (left column)
\node[xnode] (x1) at (0,3) {$1$};
\node[xnode] (x2) at (0,1) {$2$};
\node[xnode] (x3) at (0,-1) {$3$};
\node[xnode] (x4) at (0,-3) {$4$};

% Y nodes: upper right of X, same x-coordinate, "above" the Z column
\node[ynode] (y1) at (4,4.0) {$1$};
\node[ynode] (y2) at (4,3.0) {$2$};

% Z nodes: lower right of X, vertically aligned under Y
\node[znode] (z1) at (4,1.0) {$1$};
\node[znode] (z2) at (4,0.0) {$2$};
\node[znode] (z3) at (4,-1.0) {$3$};
\node[znode] (z4) at (4,-2.0) {$4$};
\node[znode] (z5) at (4,-3.0) {$5$};
\node[znode] (z6) at (4,-4.0) {$6$};

% Optional labels
\node[above=2mm] at (x1.north) {$X_1\in\{1,2,3,4\}$};
\node[above=2mm] at (y1.north) {$Y_1\in\{1,2\}$};
\node[below=2mm] at (z6.south) {$Z_1\in\{1,2,3,4,5,6\}$};

% X -> Y (upper right, blue)
\draw[->, blue, thick] (x1) -- (y1);
\draw[->, blue, thick] (x2) -- (y1);
\draw[->, blue, thick] (x3) -- (y2);
\draw[->, blue, thick] (x4) -- (y2);

% X -> Z (lower right, red)
% X=1: Z in {1,2,3}, each 1/3
\draw[->, red, thick] (x1) --  (z1);
\draw[->, red, thick] (x1) --  (z2);
\draw[->, red, thick] (x1) -- (z3);

% X=2: Z in {1,4,5}, each 1/3
\draw[->, red, thick] (x2) -- (z1);
\draw[->, red, thick] (x2) -- (z4);
\draw[->, red, thick] (x2) -- (z5);

% X=3: Z in {2,4,6}, each 1/3
\draw[->, red, thick] (x3) -- (z2);
\draw[->, red, thick] (x3) -- (z4);
\draw[->, red, thick] (x3) -- (z6);

% X=4: Z in {3,5,6}, each 1/3
\draw[->, red, thick] (x4) -- (z3);
\draw[->, red, thick] (x4) -- (z5);
\draw[->, red, thick] (x4) -- (z6);

\end{tikzpicture}}
\end{minipage}
\begin{minipage}{0.45\linewidth}
\centering
\resizebox{40mm}{!}{%
\begin{tikzpicture}[
    >=stealth,
    xnode/.style={draw, circle, minimum size=7mm},
    ynode/.style={draw, circle, minimum size=7mm},
    znode/.style={draw, circle, minimum size=7mm},
    every edge/.style={->, thick}
]

% X nodes (left column)
\node[xnode] (x1) at (0,3) {$1$};
\node[xnode] (x2) at (0,1) {$2$};
\node[xnode] (x3) at (0,-1) {$3$};
\node[xnode] (x4) at (0,-3) {$4$};

% Y nodes: upper right of X, same x-coordinate, "above" the Z column
\node[ynode] (y1) at (4,4) {$1$};
\node[ynode] (y2) at (4,3) {$2$};

% Z nodes: lower right of X, vertically aligned under Y
\node[znode] (z1) at (4,1.0) {$1$};
\node[znode] (z2) at (4,0.0) {$2$};
\node[znode] (z3) at (4,-1.0) {$3$};
\node[znode] (z4) at (4,-2.0) {$4$};
\node[znode] (z5) at (4,-3.0) {$5$};
\node[znode] (z6) at (4,-4.0) {$6$};

% Optional labels
\node[above=2mm] at (x1.north) {$X_2\in\{1,2,3,4\}$};
\node[above=2mm] at (y1.north) {$Z_2\in\{1,2\}$};
\node[below=2mm] at (z6.south) {$Y_2\in\{1,2,3,4,5,6\}$};

% X -> Y (upper right, blue)
\draw[->, red, thick] (x1) -- (y1);
\draw[->, red, thick] (x2) -- (y1);
\draw[->, red, thick] (x3) -- (y2);
\draw[->, red, thick] (x4) -- (y2);

% X -> Z (lower right, red)
% X=1: Z in {1,2,3}, each 1/3
\draw[->, blue, thick] (x1) --  (z1);
\draw[->, blue, thick] (x1) --  (z2);
\draw[->, blue, thick] (x1) -- (z3);

% X=2: Z in {1,4,5}, each 1/3
\draw[->, blue, thick] (x2) -- (z1);
\draw[->, blue, thick] (x2) -- (z4);
\draw[->, blue, thick] (x2) -- (z5);

% X=3: Z in {2,4,6}, each 1/3
\draw[->, blue, thick] (x3) -- (z2);
\draw[->, blue, thick] (x3) -- (z4);
\draw[->, blue, thick] (x3) -- (z6);

% X=4: Z in {3,5,6}, each 1/3
\draw[->, blue, thick] (x4) -- (z3);
\draw[->, blue, thick] (x4) -- (z5);
\draw[->, blue, thick] (x4) -- (z6);

\end{tikzpicture}}
\end{minipage}
\end{center}
\caption{A reversely semi-deterministic product broadcast channel.\\ \textnormal{Remark: The published version of this paper contains an
  incorrect channel diagram. The calculations are correct and correspond
  to the channel shown here. The authors apologize for this oversight.}}
\label{fif:f1}
\end{figure}
\end{claim}
\begin{proof}
We begin by showing that Marton's sum rate (the optimal sum rate) is given by $\frac 83$. Claim \ref{cl:reverselysemideterministic} shows that the sum rate of $\qmf_1\times \qmf_2$ is \begin{align*}&\min_{\lambda\in[0,1]}\big(\la \dash SR_M(\qmf_1)+\la \dash SR_M(\qmf_2)\big).\end{align*}
The result of Appendix \ref{Apendix0.5} implies that for any $\lambda\in[0,1]$, $\la \dash SR_M(\qmf_1)$ is equal to $\la \dash SR_M(\qmf_1,u(x_1))$ where $u$ is the uniform distribution on $\mathcal{X}_1$. A similar statement holds for $\la \dash SR_M(\qmf_2)$. Therefore the sum rate of $\qmf_1\times \qmf_2$ is equal to
\begin{align}&\min_{\lambda\in[0,1]}\big(\la \dash SR_M(\qmf_1,u(x_1))+\la \dash SR_M(\qmf_2,u(x_2))\big).  \label{eq:onq}\end{align}
By symmetry, $\la \dash SR_M(\qmf_2,u(x_2))=(1-\la) \dash SR_M(\qmf_1,u(x_1))$. Therefore we can express the sum rate as \begin{align*}&
\min_{\lambda\in[0,1]}\big(\la \dash SR_M(\qmf_1,u(x_1))+(1-\la) \dash SR_M(\qmf_1,u(x_1))\big).\end{align*}
In Appendix \ref{Apendix1} we show that $\la \dash SR_M(\qmf_1,u(x_1))$ is equal to
$$\la \dash SR_M(\qmf_1,u(x_1)) = \begin{cases} \begin{array}{ll} \frac 53 - \frac 23 \la & \la \in [0,\frac 12] \\ \frac 43 & \la \in [\frac 12, 1] \end{array} \end{cases}. $$
Substituting this function into \eqref{eq:onq} we see that the minimum occurs uniquely at $\lambda=0.5$ and the optimum sum rate is equal to $\frac{8}{3}$.

To compute a lower bound on the $UV$ sum rate, let $p(x_1,x_2)=u(x_1)u(x_2)$, i.e. independent uniform distribution on $\mathcal{X}_1$ and $\mathcal{X}_2$.
We define $U_1,V_1,X_1,U_2,V_2,X_2$ having a joint distribution of the form $p(u_1,v_1,x_1)p(u_2,v_2,x_2)$ as follows. Let $U_1=Y_1$ and $p(u_2,x_2)$  satisfy
\begin{align*}&\P(X_2 = 1|U_2=1) = \P(X_2 = 3|U_2=1) = \frac 12,~\mbox{and}~\P(X_2 = 2|U_2=1) = \P(X_2 = 4|U_2=1) = \frac 12,\\&
\P(U_2=1)=\P(U_2=2) =\frac 12.\end{align*} Similarly, let $V_2=Z_2$ and $p(v_1,x_1)$  satisfy
\begin{align*}&\P(X_1 = 1|V_1=1) = \P(X_1 = 3|V_1=1) = \frac 12,~\mbox{and}~\P(X_1 = 2|V_1=1) = \P(X_1 = 4|V_1=1) = \frac 12,\\&\P(V_1=1)=\P(V_1=2) =\frac 12.\end{align*}

Let  binary random variables $Q_1$ and $Q_2$ be mutually independent of each other, and independent of $U_1,V_1,X_1,U_2,V_2,X_2$. Furthermore assume that $\P(Q_1=0)=\P(Q_2=0) = \frac 45$. Define $V'_1$ and $U'_2$ as follows: When $Q_1=0$ set $V'_1 = V_1$ and else set $V'_1 = X_1$. When $Q_2=0$ set $U'_2 = U_2$ and else set $U'_2 = X_2$.
Lastly set
$\tilde{V}_1 = (V'_1, Q_1)$
$\tilde{U}_2 = (U'_2, Q_2)$.

We consider the $UV$ region for the choice of $(U_1,\tilde{U}_2)$, $(\tilde{V}_1,V_2)$, $(X_1, X_2)$. Note that
\begin{align*}
R_1 &\leq I(U_1,\tilde{U}_2; Y_1,Y_2)\\&=I(U_1; Y_1)+I(\tilde{U}_2; Y_2)\\&=
H(Y_1) + \frac 45 I(U_2;Y_2) + \frac 15 I(X_2;Y_2)\\&=1 + \frac 4{5}\cdot\frac 1{3} + \frac 1{5}\cdot1 \\&= \frac{22}{15}.\end{align*}
Similarly, one can show that
\begin{align*}
R_2 &\leq I(\tilde{V}_1,V_2; Z_1,Z_2)\\&=\frac{22}{15}.\end{align*}
The sum rate constraint on $R_1+R_2$ is as follows:
\begin{align*}
R_1 + R_2 &\leq  I(U_1,\tilde{U}_2; Y_1,Y_2)  + I(X_1,X_2; Z_1,Z_2|U_1,\tilde{U}_2)\\&=I(U_1; Y_1)  + I(X_1; Z_1|U_1)
+I(\tilde{U}_2; Y_2)  + I(X_2; Z_2|\tilde{U}_2)\\&=
H(Y_1)+I(X_1; Z_1|Y_1) +\frac 45 I(U_2;Y_2)+ \frac 15 I(X_2;Y_2)  + \frac 45 H(Z_2|U_2)\\&=
1+\frac 2{3}+\frac 4{5}\cdot\frac 1{3} + \frac 1{5}\cdot1+\frac 45\cdot1\\&=\frac{44}{15}.
\end{align*}
Similarly, one can show that
\begin{align*}
R_1 + R_2 &\leq I(\tilde{V}_1,V_2; Z_1,Z_2)+ I(X_1,X_2; Y_1,Y_2|\tilde{V}_1,V_2)\\&=\frac{44}{15}.
\end{align*}
Therefore the point $(R_1,R_2)=(\frac{22}{15},\frac{22}{15})$ is in this region. Hence the $UV$ sum rate is at least $\frac{44}{15}=3-\frac{1}{15}$. Thus for the product channel under consideration
$$ \frac 83 = SR_{M}(\qmf_1 \times \qmf_2) = SR^*(\qmf_1 \times \qmf_2) < \frac{44}{15} \leq SR_{UV}(\qmf_1 \times \qmf_2). $$
This shows that the UV outer bound is strictly suboptimal in general.
\end{proof}

\section{Capacity regions for classes of product broadcast channels}
\label{sec:capacityregions}
In this section we establish the capacity region for  some classes of product broadcast channels. Here we consider a more general setting where in addition to the private messages, the receivers also wish to decode a common message $M_0$. Hence we are interested in the achievable rate triples $(R_0, R_1, R_2)$. The capacity region is defined in a similar fashion as in the case without common message.

\subsection{An outer bound for product channels}
We present a new outer bound for the product of two broadcast channels. The manipulations here are inspired by the manipulations in the proof of Theorem \ref{th:1suff}.  This outer bound matches the capacity region for a variety of product channels, including  the product of two reversely semideterministic and  the product of two reversely more-capable channels. Hence, from Claim \ref{cl:uvsub}, it follows that  this is a {\em strictly} better bound for product broadcast channels than the UVW outer bound \cite{nai11}.

\begin{claim} \label{cl:obp} Given a product channel $\qmf(y_1,y_2,z_1,z_2|x_1,x_2)=\qmf_1(y_1,z_1|x_1)\qmf_2(y_2,z_2|x_2)$, the union over all $p_1(w_1, u_1, v_1, x_1)p_2(w_2, u_2, v_2, x_2)$ of triples $(R_0, R_1, R_2)$ satisfying
{\small \begin{align*}
R_0 & \leq \min \{ I(W_1; Y_1) + I(W_2; Y_2), I(W_1; Z_1) + I(W_2; Z_2) \} \\
R_0 + R_1 & \leq \min \{ I(W_1; Y_1) + I(W_2; Y_2), I(W_1; Z_1) + I(W_2; Z_2) \}+I(U_1; Y_1|W_1) + I(U_2; Y_2|W_2) \\
R_0 + R_2 & \leq \min \{ I(W_1; Y_1) + I(W_2; Y_2), I(W_1; Z_1) + I(W_2; Z_2) \}+I(V_1; Z_1|W_1) + I(V_2; Z_2|W_2)\\
R_0 + R_1 + R_2 & \leq \min \{ I(W_1; Y_1) + I(W_2; Y_2), I(W_1; Z_1) + I(W_2; Z_2) \} \\&\quad + I(U_2;Y_2|W_2) + I(X_2;Z_2|U_2, W_2) \\
& \quad + \min \big\{ I(U_1;Y_1|W_1) + I(X_1; Z_1|U_1, W_1), I(V_1;Z_1|W_1) + I(X_1; Y_1|V_1, W_1) \big\},\\
R_0 + R_1 + R_2 & \leq \min \{ I(W_1; Y_1) + I(W_2; Y_2), I(W_1; Z_1) + I(W_2; Z_2) \} \\&\quad+ \min \big\{ I(U_2;Y_2|W_2) + I(X_2; Z_2|U_2, W_2), I(V_2;Z_2|W_2) + I(X_2; Y_2|V_2, W_2) \big\} \\
& \quad + I(V_1;Z_1|W_1) + I(X_1;Y_1|V_1, W_1),
\end{align*}}
forms an outer bound to the capacity region of the product broadcast channel.
\end{claim}
\begin{remark}
Note that setting $X_2, Y_2, Z_2 = \emptyset$ reduces this bound to UVW outer bound. Additionally, 
one can interchange the roles of $Y_2$ and $Z_1$ with $Z_2$ and $Y_1$ respectively to get another set of similar constraints. These constraints will be over different auxiliaries distributed as  $p_1(\wt_1, \ut_1, \vt_1, x_1)p_2(\wt_2, \ut_2, \vt_2, x_2)$ (observe that the distributions on $X_1, X_2$ are preserved), and we can take the intersection of these two constraints. Finally one can take union of these two sets of constraints over all \\
$p_1(w_1, u_1, v_1,\wt_1, \ut_1, \vt_1, x_1) p_2(w_2, u_2, v_2,\wt_2, \ut_2, \vt_2,x_2) $ to get another, possibly better, outer bound.
 \end{remark}

\begin{proof}
The proof of this claim is given in  Appendix \ref{sec:pfob}.
\end{proof}

\begin{remark}
The above outer bound is {\em also strictly sub-optimal}. To see this first note that when one of the product channels is trivial, this outer bound does not give us anything beyond the UVW-outer bound \cite{nai11}. Now, consider a product of three channels, first one is trivial, the collection of two and three forms a reversely semi-deterministic pair. The new outer bound reduces to the UVW bound on the reversely semi-deterministic, and therefore it is strictly sub-optimal. However, one could argue that in order to write the outer bound, one should take the intersection of all possible outer bounds one can write by breaking up the broadcast channel into product forms. To deal with this objection one can consider the product of three channels as above and then slightly perturb the channel to destroy the product form structure of the channel. Because the above outer bound is continuous in the underlying channel, this outer bound must be loose for this channel. In fact, we still don't know of ``the correct way" to write an outer bound that fully captures the spirit of the counterexample discussed earlier. We have thought of alternative expressions but none seemed satisfactory.
\end{remark}

\subsubsection{An achievable region for a product broadcast channel}
\label{sse:mibp}
Given a product channel $\qmf(y_1,y_2,z_1,z_2|x_1,x_2)=\qmf_1(y_1,z_1|x_1)\qmf_2(y_2,z_2|x_2)$ the union of rate triples satisfying
{\small \begin{align}
R_0 & \leq \min \{ I(W_1; Y_1) + I(W_2; Y_2), I(W_1; Z_1) + I(W_2; Z_2) \}  \nonumber \\
R_0 + R_1 & \leq I(W_1; Y_1) + I(W_2; Y_2) +I(U_1; Y_1|W_1) + I(U_2; Y_2|W_2) \label{eq:mibp}\\
R_0 + R_2 & \leq I(W_1; Z_1) + I(W_2; Z_2) +I(V_1; Z_1|W_1) + I(V_2; Z_2|W_2) \nonumber\\
R_0 + R_1 + R_2 & \leq \min \{ I(W_1; Y_1) + I(W_2; Y_2), I(W_1; Z_1) + I(W_2; Z_2) \} \nonumber\\
&\quad +I(U_1; Y_1|W_1) + I(U_2; Y_2|W_2)+I(V_1; Z_1|W_1) \nonumber \\
& \quad + I(V_2; Z_2|W_2)-
I(U_1; V_1|W_1) - I(U_2; V_2|W_2) \nonumber
\end{align}}\\
over all $p_1(w_1, v_1, u_1, x_1)p_2(w_2, v_2, u_2, x_2)$ constitutes an inner bound to the capacity region. The achievability of these points   is immediate from Marton's inner bound by letting $U=(U_1, U_2), V=(V_1, V_2), W=(W_1, W_2)$ and $p(u,v,w) \sim p_1(w_1, v_1, u_1, x_1)p_2(w_2, v_2, u_2, x_2)$.

\subsection{Capacity regions for new classes of  product broadcast channels}
\begin{theorem} \label{Thm2} The capacity region for a product of reversely semi-deterministic (say, channels $X_1 \to Y_1, X_2 \to Z_2$ are deterministic) broadcast channel is given by
the union of rate triples satisfying
{\small \begin{align*}
R_0 & \leq \min \{ I(W_1; Y_1) + I(W_2; Y_2), I(W_1; Z_1) + I(W_2; Z_2) \}   \\
R_0 + R_1 & \leq I(W_1; Y_1) + I(W_2; Y_2) +H(Y_1|W_1) + I(U_2; Y_2|W_2) \\
R_0 + R_2 & \leq I(W_1; Z_1) + I(W_2; Z_2) +I(V_1; Z_1|W_1) + H(Z_2|W_2) \\
R_0 + R_1 + R_2 & \leq \min \{ I(W_1; Y_1) + I(W_2; Y_2), I(W_1; Z_1) + I(W_2; Z_2) \} \\
&\quad +I(V_1; Z_1|W_1) + H(Y_1|V_1, W_1) + I(U_2;Y_2|W_2) + H(Z_2|U_2, W_2)
\end{align*}}
over all $p_1(w_1, v_1,  x_1)p_2(w_2,  u_2, x_2).$
\end{theorem}

\begin{proof}
The achievability is immediate by setting $U_1=Y_1$ and $V_2=Z_2$ in \eqref{eq:mibp}. Note that these two choices of auxiliary random variables are possible since channels $X_1 \to Y_1, X_2 \to Z_2$ are deterministic.

The converse is also immediate from the outer bound in  Claim \ref{cl:obp}. Observe that for any $p_1(w_1, v_1, u_1, x_1)$, $p_2(w_2, v_2, u_2, x_2)$ we have
$$I(U_1; Y_1|W_1) \leq H(Y_1|W_1), ~~ I(V_2; Z_2|W_2) \leq H(Z_2|W_2),$$ and each of the two sum rate terms  in Claim~\ref{cl:obp} is bounded by
\begin{align*}
& \min \{ I(W_1; Y_1) + I(W_2; Y_2), I(W_1; Z_1) + I(W_2; Z_2) \} +I(V_1; Z_1|W_1)\\
&\quad  + H(Y_1|V_1, W_1) + I(U_2;Y_2|W_2) + H(Z_2|U_2, W_2).
\end{align*}
Thus the outer bound is contained in the inner bound (and hence they coincide).
\end{proof}

\begin{theorem} \label{th:mcp} The capacity region for a product of reversely more-capable (say, receiver $Z_1$ is more capable than $Y_1$,  and receiver $Y_2$ is more capable than $Z_2$) broadcast channel is given by
the union of rate triples satisfying
{\small \begin{align*}
R_0 & \leq \min \{ I(W_1; Y_1) + I(W_2; Y_2), I(W_1; Z_1) + I(W_2; Z_2) \} \\
R_0 + R_1 & \leq \min \{ I(W_1; Y_1) + I(W_2; Y_2), I(W_1; Z_1) + I(W_2; Z_2) \}+I(U_1; Y_1|W_1) + I(X_2; Y_2|W_2) \\
R_0 + R_2 & \leq \min \{ I(W_1; Y_1) + I(W_2; Y_2), I(W_1; Z_1) + I(W_2; Z_2) \}+I(X_1; Z_1|W_1) + I(V_2; Z_2|W_2)\\
R_0 + R_1 + R_2 & \leq \min \{ I(W_1; Y_1) + I(W_2; Y_2), I(W_1; Z_1) + I(W_2; Z_2) \} + I(X_2;Y_2|W_2)  \\
& \quad + \min \big\{ I(U_1;Y_1|W_1) + I(X_1; Z_1|U_1, W_1), I(X_1;Z_1|W_1)  \big\},\\
R_0 + R_1 + R_2 & \leq \min \{ I(W_1; Y_1) + I(W_2; Y_2), I(W_1; Z_1) + I(W_2; Z_2) \} \\&\quad+ \min \big\{ I(X_2;Y_2|W_2) , I(V_2;Z_2|W_2) + I(X_2; Y_2|V_2, W_2) \big\} + I(X_1;Z_1|W_1)
\end{align*}}
over all $p_1(w_1, v_1,  x_1)p_2(w_2,  u_2, x_2).$
\end{theorem}

\begin{proof}
The achievability is immediate by setting $W_1' = (U_1, W_1) , U_1'=\emptyset, V_1'=X_1 $ and $W_2'=(V_2, W_2), U_2' =X_2', V_2'= \emptyset$ in \eqref{eq:mibp}. Plugging these choices into \eqref{eq:mibp} we obtain that one can achieve rate triples satisfying
{\small \begin{align*}
R_0 & \leq \min \{ I(U_1,W_1; Y_1) + I(V_2, W_2; Y_2), I(U_1,W_1; Z_1) + I(V_2,W_2; Z_2) \}   \\
R_0 + R_1 & \leq I(U_1,W_1; Y_1) + I(V_2, W_2; Y_2)  + I(X_2; Y_2|V_2, W_2) \\
& = I(W_1; Y_1) + I(W_2; Y_2)  + (U_1; Y_1|W_1) + I(X_2; Y_2|W_2) \\
R_0 + R_2 & \leq I(U_1, W_1; Z_1) + I(V_2, W_2; Z_2) +I(X_1; Z_1|U_1, W_1)  \\
& = I(W_1; Z_1) + I(W_2; Z_2) +I(X_1; Z_1|W_1) + I(V_2; Z_2|W_2) \\
R_0 + R_1 + R_2 & \leq \min \{ I(U_1, W_1; Y_1) + I(V_2, W_2; Y_2), I(U_1, W_1; Z_1) + I(V_2, W_2; Z_2) \} \\
&\quad  + I(X_2; Y_2|V_2, W_2)+I(X_1; Z_1|U_1, W_1)  . \nonumber
\end{align*}}

The last sum rate term can be split into two terms as follows
{\small \begin{align*}
R_0 + R_1 + R_2 & \leq  I(U_1, W_1; Y_1) + I(V_2, W_2; Y_2) +  I(X_2; Y_2|V_2, W_2)+I(X_1; Z_1|U_1, W_1)  \\
& = I(W_1; Y_1) + I(W_2; Y_2) +  I(X_2; Y_2| W_2)+  I(U_1;Y_1|W_1) + I(X_1; Z_1|U_1, W_1)  \\
R_0 + R_1 + R_2 & \leq  I(U_1, W_1; Z_1) + I(V_2, W_2; Z_2) +  I(X_2; Y_2|V_2, W_2)+I(X_1; Z_1|U_1, W_1)  \\
& = I(W_1; Z_1) + I(W_2; Z_2) +  I(X_1; Z_1| W_1)+  I(V_2;Z_2|W_2) + I(X_2; Y_2|V_2, W_2).
\end{align*}}

Thus we see, by comparing term by term,  that this achievable region is at least as large as the region stated in Theorem \ref{th:mcp}, and hence the region in  Theorem \ref{th:mcp} is achievable.

The converse is also reasonably immediate from the outer bound in  Claim \ref{cl:obp}. Observe the following:
\begin{align*}
& \min \{ I(W_1; Y_1) + I(W_2; Y_2), I(W_1; Z_1) + I(W_2; Z_2) \}+I(U_1; Y_1|W_1) + I(U_2; Y_2|W_2) \\
& \quad \leq  \min \{ I(W_1; Y_1) + I(W_2; Y_2), I(W_1; Z_1) + I(W_2; Z_2) \}+I(U_1; Y_1|W_1) + I(X_2; Y_2|W_2), \\
& \min \{ I(W_1; Y_1) + I(W_2; Y_2), I(W_1; Z_1) + I(W_2; Z_2) \}+I(V_1; Z_1|W_1) + I(V_2; Z_2|W_2) \\
& \quad \{ I(W_1; Y_1) + I(W_2; Y_2), I(W_1; Z_1) + I(W_2; Z_2) \}+I(X_1; Z_1|W_1) + I(V_2; Z_2|W_2),  \\
& \min \{ I(W_1; Y_1) + I(W_2; Y_2), I(W_1; Z_1) + I(W_2; Z_2) \}  + I(U_2;Y_2|W_2) + I(X_2;Z_2|U_2, W_2) \\
& ~ + \min \big\{ I(U_1;Y_1|W_1) + I(X_1; Z_1|U_1, W_1), I(V_1;Z_1|W_1) + I(X_1; Y_1|V_1, W_1) \big\} \\
& \quad \leq   \min \{ I(W_1; Y_1) + I(W_2; Y_2), I(W_1; Z_1) + I(W_2; Z_2) \} + I(X_2;Y_2|W_2)  \\
& \qquad + \min \big\{ I(U_1;Y_1|W_1) + I(X_1; Z_1|U_1, W_1), I(X_1;Z_1|W_1)  \big\}, 
\end{align*}
and finally,
\begin{align*}
&  \min \{ I(W_1; Y_1) + I(W_2; Y_2), I(W_1; Z_1) + I(W_2; Z_2) \}  + I(V_1;Z_1|W_1) + I(X_1;Y_1|V_1, W_1) \\
&~ + \min \big\{ I(U_2;Y_2|W_2) + I(X_2; Z_2|U_2, W_2), I(V_2;Z_2|W_2) + I(X_2; Y_2|V_2, W_2) \big\} \\
& \quad \leq  \{ I(W_1; Y_1) + I(W_2; Y_2), I(W_1; Z_1) + I(W_2; Z_2) \}  + I(X_1;Z_1|W_1)  \\
&\qquad + \min \big\{ I(X_2; Y_2| W_2), I(V_2;Z_2|W_2) + I(X_2; Y_2| V_2, W_2) \big\}. 
\end{align*}

Thus we see, by comparing term by term,  that  the region stated in Theorem \ref{th:mcp}  is at least as large as the outer bound in Claim \ref{cl:obp}. Hence the region in  Theorem \ref{th:mcp} is an outer bound, thus completing the converse.
\end{proof}

\begin{remark}
The achievable region in \eqref{eq:mibp} also matches the outer bound in Claim \ref{cl:obp} for a variety of other classes. For instance,  say $Z_1$ is more capable than $Y_1$ and $Y_2$ is a deterministic function of $X_2$. In this case, one can show that the capacity region is given by
the union of rate triples satisfying
{\small \begin{align*}
R_0 & \leq \min \{ I(W_1; Y_1) + I(W_2; Y_2), I(W_1; Z_1) + I(W_2; Z_2) \} \\
R_0 + R_1 & \leq \min \{ I(W_1; Y_1) + I(W_2; Y_2), I(W_1; Z_1) + I(W_2; Z_2) \}+I(U_1; Y_1|W_1) + H(Y_2|W_2) \\
R_0 + R_2 & \leq \min \{ I(W_1; Y_1) + I(W_2; Y_2), I(W_1; Z_1) + I(W_2; Z_2) \}+I(X_1; Z_1|W_1) + I(V_2; Z_2|W_2)\\
R_0 + R_1 + R_2 & \leq \min \{ I(W_1; Y_1) + I(W_2; Y_2), I(W_1; Z_1) + I(W_2; Z_2) \} \\
&\quad+ I(V_2;Z_2|W_2) + H(Y_2|V_2, W_2)  + I(X_1;Z_1|W_1).
\end{align*}}
The details are left to the reader.
\end{remark}

\section{On Marton\rq{}s inner bound and $\lambda$-sum-rate}
\label{sec:gen}
In this section we prove a collection of results regarding Marton's inner bound and also about the quantity we introduced earlier, the $\la$-sum-rate.

\subsection{Two letter Marton's inner bound}
\label{sec:TwoLetter}
This section considers the two letter Marton's inner bound and the role it plays in determining the optimality of the traditional Marton's inner bound. To simplify our analysis and for the ease of exposition we will focus on the sum rate, but some of the insights that we obtained have already been useful beyond just the sum rate.

Given a broadcast channel $\qmf(y,z|x)$ the maximum sum rate achievable via Marton's strategy is given by
 \begin{align}
SR_M(\qmf)  = \max_{p(u,v,w,x)} \min \{I(W;Y), I(W;Z) \} + I(U;Y|W) + I(V;Z|W) - I(U;V|W).
\label{eq:mib1l}
\end{align}
The maximum is taken over distributions $p(u,v,w,x)$ where the auxiliary random variables  satisfy the Markov chain $(U,V,W) \mar X \mar (Y,Z)$.

Consider a product broadcast channel $\qmf(y_1,z_1|x_1) \times \qmf(y_2,z_2|x_2)$ obtained by taking identical copies of the original channel. One can obtain the maximum sum rate achievable via Marton's strategy for this new channel as
 {\small \begin{align}
 SR_M(\qmf \times \qmf) \nonumber  & = \max_{p(u,v,w,x_1,x_2)} \min \{I(W;Y_1, Y_2), I(W;Z_1,Z_2) \}  + I(U;Y_1, Y_2|W) \\
& \qquad  \qquad \qquad \qquad
 + I(V;Z_1, Z_2|W) - I(U;V|W).
\label{eq:mib2l}
	\end{align}}
Here the maximum is taken over distributions $p(u,v,w,x_1,x_2)$ where the auxiliary random variables $(U,V,W)$ satisfy the Markov chain: $(U,V,W) \mar (X_1,X_2) \mar (Y_1,Y_2,Z_1, Z_2)$, and the channel has a  product nature given by $\qmf(y_1,y_2,z_1,z_2|x_1,x_2) = \qmf(y_1,z_1|x_1) \qmf(y_2,z_2|x_2).$ Define $SR_{2M}(\qmf) := \frac{1}{2} SR_M(\qmf \times \qmf)$ to be the two-letter sum rate yielded by Marton's inner bound.

Here we state a (folk-lore) lemma that relates the optimality of Marton's achievable strategy and the relationship between $SR_{2M}(\qmf)$ and $SR_M(\qmf)$.
\begin{lemma} (Folklore)
The following two statements are equivalent:
\begin{enumerate}
\item Marton's achievable strategy achieves the optimal sum rate, $SR^*(\qmf)$, for all broadcast channels $\qmf(y,z|x)$, i.e. $SR_M(\qmf)=SR^*(\qmf)$.
\item $SR_{2M}(\qmf) = SR_M(\qmf)$ for all $\qmf(y,z|x)$.
\end{enumerate}
\label{le:2to1mib}
\end{lemma}
\begin{proof} We present an argument here for completeness.

(1 $\implies$ 2) This follows from two facts: first, $SR_{2M}(\qmf)$ yields an achievable sum rate for the broadcast channel $\qmf(y,z|x)$, i.e. $SR_{2M}(\qmf) \leq SR^*(\qmf)$; and second, $SR_{2M}(\qmf) \geq SR_M(\qmf)$ for all $\qmf(y,z|x)$.
To see the first, observe that a codebook of block length $n$ for the product channel $\qmf(y_1,z_1|x_1) \qmf(y_2,z_2|x_2)$ yields a codebook of block length $2n$ for the original channel $\qmf(y,z|x)$, since the mapping from $(x_1,...,x_{2n})$ to the pairs $(y_1, \ldots y_{2n})$,  $(z_1, \ldots, z_{2n})$ by the channel $\qmf(y,z|x)$ is same as the mapping from $((x_1,x_2),....,(x_{2n-1},x_{2n}))$ to the pairs $((y_1,y_2),....,(y_{2n-1},y_{2n}))$, and $((z_1,z_2),....,(z_{2n-1},z_{2n}))$ by the channel $\qmf(y_1,z_1|x_1) \qmf(y_2,z_2|x_2)$. Hence any rate achievable for the product channel $\qmf(y_1,z_1|x_1) \qmf(y_2,z_2|x_2)$ (normalized by factor $\frac 12$) is also achievable for the single channel $\qmf(y,z|x)$.

Let $p^*(u,v,w,x)$ achieve the maximum sum rate in \eqref{eq:mib1l}. Choose $\Ut=(U_1,U_2), \Vt=(V_1, V_2), \Wt=(W_1,W_2)$ and let $p(\ut,\vt,\wt,x_1,x_2) = p^*(u_1, v_1, w_1, x_1) p^*(u_2, v_2, w_2, x_2)$, i.e. take a product distribution by taking two i.i.d. copies of the single letter optimal distribution. Now observe that
{\small \begin{align*}
 2 SR_{2M}(\qmf) & \geq \min \{I(\Wt;Y_1, Y_2), I(\Wt;Z_1,Z_2) \} + I(\Ut;Y_1, Y_2|\Wt) + I(\Vt;Z_1, Z_2|\Wt) - I(\Ut;\Vt|\Wt)\\
& =  \min \{I(W_1;Y_1), I(W_1;Z_1) \} + I(U_1;Y_1|W_1) \quad + I(V_1;Z_1|W_1) - I(U_1;V_1|W_1) \\
&\quad + \min \{I(W_2;Y_2), I(W_2;Z_2) \} + I(U_2;Y_2|W_2) + I(V_2;Z_2|W_2) - I(U_2;V_2|W_2) \\
&  = 2 SR_M(\qmf).
\end{align*}}
This shows that if  $SR_M(\qmf)$ is the maximum achievable sum rate then  $SR_{2M}(\qmf) = SR_M(\qmf)$ for all $\qmf(y,z|x)$.

\medskip

(2 $\implies$ 1) Let $\qmf \otimes_n (y_1^n, z_1^n|x_1^n) = \prod_{i=1}^n \qmf(y_i,z_i|x_i)$ denote the $n$-fold product channel.
 If  2 holds then, by induction,  for any $k \geq 1$ the $2^k$-fold product channel satisfies
$$ \frac{1}{2^k} SR_M(\qmf\otimes_{2^k}) = SR_M(\qmf).$$
However for any $n$, we know from Fano's inequality that for any sequence of good codebooks
\begin{align*}
& n(R_1 + R_2) \\
& \quad \leq I(M_1;Y_{1}^n) + I(M_2;Z_{1}^n) + n (R_1 + R_2) \e_n + 1\\
&\quad \leq SR_M(\qmf \otimes_n ) + n(R_1 + R_2) \e_n + 1.
\end{align*}
where $SR_M(\qmf \otimes_n)$ is the maximum sum rate by Marton's strategy for the $n$-fold product channel, as setting $U=M_1, V=M_2, W=\emptyset$ is a particular choice of the auxiliary random variables for the $n$-fold product channel. Further we also know that $\e_n \to 0$ as $n \to \infty$. This implies that the optimal sum rate, $SR^*(\qmf)$, for the broadcast channel $\qmf(y,z|x)$ satisfies
$$ SR^*(\qmf) \leq \liminf_n \frac 1n SR_M(\qmf \otimes_n) \leq \lim_{k \to \infty}  \frac{1}{2^k} SR_M(\qmf \otimes_{2^k}) = SR_M(\qmf). $$
On the other hand $SR_M(\qmf) \leq SR^*(\qmf)$ since $SR_M(\qmf)$ is the rate given by Marton's achievable strategy. Hence we have $SR_M(\qmf) = SR^*(\qmf)$.
\end{proof}

\begin{remark}
Lemma \ref{le:2to1mib} is an attempt at answering the question of whether Marton'd inner bound is optimal or not. If one can find a channel for which $SR_{2M}(\qmf) > SR_M(\qmf)$ then Marton's inner bound is strictly sub-optimal, otherwise (i.e. for all channels $\qmf$ we have $SR_{2M}(\qmf) = SR_M(\qmf)$) Marton's inner bound is optimal  for the sum rate and would yield the capacity region. The advantage of just having to look at 2-letter extensions is that with the recently established cardinality bounds \cite{goa09} one can numerically search over channels $\qmf$, to try and determine a channel where $SR_{2M}(\qmf) > SR_M(\qmf)$. So far, our searches have yielded evidence to the contrary, i.e. they point towards a potential optimality of Marton's coding scheme.
\end{remark}

\subsection{Properties of $\la$-sum-rate}
\label{subsec:lambdasumrate}
In this section, we state some results about the $\la$-sum-rate (defined in Section \ref{sse:prelimd}) as this quantity seems to possess properties (such as factorizations over $\qmf_1 \times \qmf_2$) which we will show that $SR_M(\qmf)$ does not possess. Further $\la$-sum-rate also gives us a lot of insight into evaluations of the various bounds and in the search for potential counterexamples to optimality of Marton.
\begin{lemma}
\label{le:conla}
For a given channel $\qmf(y,z|x)$, $\la\dash SR_M(\qmf)$ and $\la\dash SR_M(\qmf,p(x))$ are convex in $\la$ for $\la \in [0,1]$, and concave in $p(x)$ for a fixed $\lambda \in [0,1]$.
\end{lemma}
\begin{proof} To show that $\lambda\mapsto \la\dash SR_M(\qmf, p(x))$ is convex, take arbitrary $\lambda_1, \lambda_2, \lambda_3$ satisfying $\lambda_2=\frac{\lambda_1+\lambda_3}{2}$. Take some $p(w^*,u^*,v^*|x)$ maximizing $\la_2\dash SR_M(\qmf, p(x))$. Note that
\begin{align*}&\la_2\dash SR_M(\qmf, p(x))=\\&\big\{\lambda_2 I(W^*;Y)+(1-\lambda_2)I(W^*;Z)+I(U^*;Y|W^*)+I(V^*;Z|W^*)-I(U^*;V^*|W^*)\big\}=\\&
\frac{1}{2}\bigg[\big\{\lambda_1 I(W^*;Y)+(1-\lambda_1)I(W^*;Z)+I(U^*;Y|W^*)+I(V^*;Z|W^*)-I(U^*;V^*|W^*)\big\}+
\\&\big\{\lambda_3 I(W^*;Y)+(1-\lambda_3)I(W^*;Z)+I(U^*;Y|W^*)+I(V^*;Z|W^*)-I(U^*;V^*|W^*)\big\}\bigg]\leq\\&
\frac{1}{2}\big[\la_1\dash SR_M(\qmf, p(x))+\la_3\dash SR_M(\qmf, p(x))\big].\end{align*}
To show that $\lambda\mapsto \la\dash SR_M(\qmf)$ is convex, let $p^*(x)$ be the maximizing input distribution, i.e. $\la\dash SR_M(\qmf,p^*(x))=\la\dash SR_M(\qmf)$. Note that
\begin{align*}
\la\dash SR_M(\qmf) & =\la\dash SR_M(\qmf,p^*(x))\\
& \leq
\frac{1}{2}\big[\la_1\dash SR_M(\qmf,p^*(x))+\la_3\dash SR_M(\qmf,p^*(x))\big]\\
&\leq
\max_{p(x)}\frac{1}{2}\la_1\dash SR_M(\qmf,p(x))+\max_{p(x)}\frac{1}{2}\la_3\dash SR_M(\qmf,p(x))\\
&=
\frac{1}{2}\big[\la_1\dash SR_M(\qmf)+\la_3\dash SR_M(\qmf)\big].\end{align*}

To show the concavity in $p(x)$ take two marginal distributions $p_0(x)$ and $p_1(x)$, and assume that $(U_0,V_0,W_0,X_0)$ and $(U_1,V_1,W_1,X_1)$ are two set of random variables maximizing the expressions of $\la\dash SR_M(\qmf,p_0(x))$ and $\la\dash SR_M(\qmf,p_1(x))$ respectively. Take a uniform binary random variable $Q$, independent of all previously defined random variables. Let $U=U_Q$, $V=V_Q$, $W=(W_Q,Q)$, $X=X_Q$. Observe that $X$ is distributed according to $\frac{p_0(x)}{2}+\frac{p_1(x)}{2}$ and
\begin{align*}
\la\dash SR_M(\qmf,\frac{p_0(x)}{2}+\frac{p_1(x)}{2}) & \geq \la I(W_Q,Q;Y) + \bar\la I(W_Q,Q;Z) + I(U_Q;Y|W_Q,Q) \\
& \quad + I(V_Q;Z|Q,W_Q) - I(U_Q;V_Q|W_Q,Q)\\
& = \la I(Q;Y) + \bar \la I(Q;Z) + \frac 12 \la\dash SR_M(\qmf,p_0(x)) + \frac 12 \la\dash SR_M(\qmf,p_1(x)).
\end{align*}
Thus, $\la\dash SR_M(\qmf,p(x))$ is concave in $p(x)$.
\end{proof}

\begin{lemma}
$\la \dash SR_M(\qmf)$ is related to the optimal sum rate as follows:
$$ \min_{\la \in \{0,1\}} \la \dash SR_M(\qmf) \geq SR^*(\qmf), $$
i.e. the minimum value of $\la\dash SR_M(\qmf)$ at the boundaries, i.e. $\la =0,1$, yields an upper bound on the optimal sum rate, $SR^*(\qmf)$ .
\end{lemma}
\begin{proof} We prove the statement for $\la=0$; the proof for $\la=1$ is similar. We will show that 
\begin{equation}
0\dash SR_M(\qmf, p(x))=\max_{p(w|x)}I(W;Z)+I(X;Y|W).
\label{eq:0lsr}
\end{equation}
Once \eqref{eq:0lsr} is established the proof then becomes an immediate, as $0\dash SR_M(\qmf, p(x))$ will in turn be an upper bound on the optimal sum rate by the UV outer bound (replace W by V).

  To show \eqref{eq:0lsr} first note that by setting $V = \emptyset, U = X$ in \eqref{eq:lmib1lm} we obtain
$$0\dash SR_M(\qmf, p(x)) \geq \max_{p(w|x)}I(W;Z)+I(X;Y|W).$$

To obtain the other direction, observe that
\begin{align}0\dash SR_M(\qmf, p(x))&=\max_{p(u,v,w|x)}\big\{I(W;Z)+I(U;Y|W)+I(V;Z|W)-I(U;V|W)\big\}\nonumber \\&=
\max_{p(u,v,w|x)}\big\{I(VW;Z)+I(U;Y|W)-I(U;V|W)\big\}\nonumber\\
&= \max_{p(u,v,w|x)}\big\{I(VW;Z)+I(U;Y|VW) - I(U;V|WY)\big\}\nonumber\\
& \leq \max_{p(u,v,w|x)}\big\{I(VW;Z)+I(X;Y|VW)\big\}\nonumber\\
&=\max_{p(w'|x)}I(W';Z)+I(X;Y|W').\nonumber\end{align}
Note that in the last step we replace $(V,W)$ by $W'$.

Thus we have, as desired,
$$0\dash SR_M(\qmf, p(x)) = \max_{p(w|x)}I(W;Z)+I(X;Y|W).$$
\end{proof}
\begin{corollary}
\label{co:obs}
If the minimum value of $\la \dash SR_M(\qmf)$ is attained at $\la=0$ or $\la=1$ then $SR_M(\qmf) = SR^*(\qmf)$, i.e. Marton's strategy achieves the optimal sum rate.
\end{corollary}

\begin{proof}
This follows from the relationships
$$ \min_{\la \in [0,1]}  \la \dash SR_M(\qmf) = SR_M(\qmf) \leq SR^*(\qmf) \leq \min_{\la \in \{0,1\}} \la \dash SR_M(\qmf). $$
\end{proof}

\begin{lemma}
\label{le:cardlsr}
To compute the maximum sum rate in \eqref{eq:lmib1l}, it suffices to consider auxiliary random variables that satisfy  $|\Uc|, |\Vc|, |\Wc| \leq |\Xc|$.
\end{lemma}
\begin{proof} This is proved in Theorem 2 of \cite{gea10}.
\end{proof}

\begin{lemma}\label{claim:derivative} Take some arbitrary $p(x)$ and real $\lambda^*$. Then for any $p(w^*,u^*,v^*|x)$ achieving $\la^*\dash SR_M(\qmf,p(x))$, the line $\lambda\mapsto (\lambda-\lambda^*)(I(W^*;Y)-I(W^*;Z))+\la^*\dash SR_M(\qmf,p(x))$ is a supporting hyperplane to the convex curve $\lambda\mapsto \la\dash SR_M(\qmf,p(x))$.
\end{lemma}

\begin{proof} At $\lambda=\lambda^*$, the expression $(\lambda-\lambda^*)(I(W^*;Y)-I(W^*;Z))+\la^*\dash SR_M(\qmf,p(x))$ is equal to $\la^*\dash SR_M(\qmf,p(x))$ which is a point on the curve $\lambda\mapsto \la\dash SR_M(\qmf,p(x))$. We need to show that for any arbitrary $\lambda$,
\begin{align*}&\la\dash SR_M(\qmf,p(x))\geq (\lambda-\lambda^*)(I(W^*;Y)-I(W^*;Z))+\la^*\dash SR_M(\qmf,p(x)).\end{align*}
The above inequality holds because it is equivalent to
\begin{align*}&\la\dash SR_M(\qmf,p(x))\geq \lambda I(W^*;Y)+(1-\lambda)I(W^*;Z)+I(U^*;Y|W^*)+I(V^*;Z|W^*)-I(U^*;V^*|W^*).\end{align*}
\end{proof}

\begin{lemma} $\la \dash SR_M(\qmf,p(x))$ is linear in $\lambda$ for more capable\footnote{For the definitions of less noisy broadcast channel or more capable broadcast channel, please refer to \cite{elk11}.} channels and constant in $\lambda$ for less noisy channels and deterministic channels.
\end{lemma}

\begin{proof} 
More capable: Assume that $Y$ is more capable than $Z$.
\begin{align*}
\la \dash SR_M(\qmf,p(x)) &= \max_{p(u,v,w|x)}\big[\lambda I(W;Y) + (1-\lambda) I(W;Z) + I(U;Y|W) + I(V;Z|W) - I(U;V|W)\big] \\
& \leq \max_{p(u,v,w|x)}\big[\lambda I(W;Y) + (1-\lambda) I(W;Z) +I(U;Y|W) +  I(V;Z|U,W)\big] \\
& \leq \max_{p(u,w|x)}\big[\lambda I(W;Y) + (1-\lambda) I(W;Z) +I(U;Y|W) +  I(X;Z|U,W)\big] \\
& \stackrel{(a)}{\leq} \max_{p(u,w|x)}\big[\lambda I(W;Y) + (1-\lambda) I(W;Z) +I(U;Y|W) +  I(X;Y|U,W)\big] \\
&  \leq \max_{p(w|x)}\big[\lambda I(W;Y) + (1-\lambda) I(W;Z) + I(X;Y|W)\big]\\
& = I(X;Y) + (1-\lambda) \max_{p(w|x)}(I(W;Z) - I(W;Y)).
\end{align*}
The inequality marked $(a)$ is justified by the more-capable assumption.
On the other hand setting $U=X$, $V=\emptyset$ shows that $\la \dash SR_M(\qmf,p(x))\geq I(X;Y) + (1-\lambda) \max_{p(w|x)}(I(W;Z) - I(W;Y))$.
Thus, when  $Y$ is more capable than $Z$,
\begin{equation}
\la \dash SR_M(\qmf,p(x)) =  I(X;Y) + (1-\lambda) \max_{p(w|x)}(I(W;Z) - I(W;Y))\label{eq:lmc}
\end{equation}
and  is linear in $\lambda$.

Less Noisy: Assume that $Y$ is less noisy than $Z$; and hence   $Y$ is also more capable than $Z$. From \eqref{eq:lmc}
\begin{align*}
\la \dash SR_M(\qmf,p(x)) &=  I(X;Y) + (1-\lambda) \max_{p(w|x)}(I(W;Z) - I(W;Y))\\
& = I(X;Y).
\end{align*}
The second equality follows since $I(W;Z) \leq I(W;Y),~ \forall~ W \to X \to (Y,Z)$ (definition of less noisy) and equality can be achieved by setting $W=\emptyset$.

Deterministic:
\begin{align*}
\la \dash SR_M(\qmf,p(x)) &= \max_{p(u,v,w|x)}\big[\lambda I(W;Y) + (1-\lambda) I(W;Z) + I(U;Y|W) + I(V;Z|W) - I(U;V|W)\big] \\
& \leq \max_{p(u,v,w|x)}\big[\lambda I(W;Y) + (1-\lambda) I(W;Z) +I(U;Y|W) +  I(V;Z|U,W)\big] \\
& \leq \max_{p(u,v,w|x)}\big[\lambda I(W;Y,Z) + (1-\lambda) I(W;Y,Z) +I(U;Y,Z|W) +  I(V;Y,Z|U,W)\big] \\
&  \leq I(X;Y,Z) = H(Y,Z).
\end{align*}
One the other hand setting $W=\emptyset$, $U=Y$, $V=Z$ shows that $\la \dash SR_M(\qmf,p(x))\geq H(Y,Z)$. Hence $\la \dash SR_M(\qmf,p(x))$ is a constant. (Note that these choices of auxiliaries, i.e. $U=Y$, $V=Z$, are permissible for deterministic channels since $(U,V) \to X \to (Y,Z)$ is a  Markov chain.)

\end{proof}

\begin{remark}
In each of the cases above, it is clear that the minimizing $\lambda$ for the $\la\dash SR_M(\qmf)$ lies  on $\la \in \{0,1\}$. Thus the optimality of $SR_M$ could be deduced alternately using Corollary \ref{co:obs}.
\end{remark}

\subsection{On $SR_M$ for product channels}
\label{subsec:lambdasumrateproduct}
In this section we consider the behavior of $SR_M(\qmf)$ for the product of two non-identical channels. 
An interested reader may wonder why we considered {\em factorization of $\la\dash SR_M(\qmf)$} as opposed to {\em factorization of $SR_M(\qmf)$}. Indeed we will show that there are channels $\qmf_1, \qmf_2$ such that
$$  SR_M(\qmf_1 \times \qmf_2) > SR_M(\qmf_1) + SR_M(\qmf_2). $$

\begin{lemma}
\label{le:mnf}
Let $p = 0.1, e = H(0.1) = \log_2 10 - 0.9 \log_2 9$. Consider a product channel formed by the following components: Let the channels $X_1 \to Y_1$ and $X_2 \to Z_2$ be $BEC(e)$ and the channels $X_1 \to Z_1$ and $X_2 \to Y_2$ be $BSC(p)$. For this product channel
$$ SR_M  (\qmf_1 \times \qmf_2) > SR_M(\qmf_1) + SR_M(\qmf_2). $$
\label{le:nofmib}
\end{lemma}
\begin{proof}
From \cite{nai10}, since $1-e = 1 - H(p)$ we know that $Y_1$ is more capable than $Z_1$ and $Z_2$ is more capable than $Y_2$. Thus, from Theorem \ref{th:mcp}, we know that Marton's inner bound is optimal for this channel. Hence from Lemma \ref{le:mm} and Theorem \ref{th:fact}, we have that
$$ SR_M  (\qmf_1 \times \qmf_2) = \min_{\la \in [0,1]} \la \dash SR_M(\qmf_1) + \la \dash SR_M(\qmf_2). $$
By the skew-symmetry we know that $\la \dash SR_M(\qmf_2) = (1-\la) \dash SR_M(\qmf_1)$. Further, from the symmetry, it is easy to show that it suffices to consider $\P(X=0) = \frac 12$ to compute $\la \dash SR_M(\qmf_1)$. In particular one can show that
$$ \la \dash SR_M(\qmf_1) = C + (1-\la) d^*, $$
where $C$ is the common capacity of the $BSC(p)$ and $BEC(e)$, and $d^* = \max_{p(x)} I(X;Y) - I(X;Z)$. For the chosen parameters $d^* \approx 0.03877$.
The maximum sum rate of the channel $\qmf_1(y_1, z_1|x_1)$, since $Y_1$ is more capable than $Z_1$, is given by the capacity to receiver $Y_1$;  hence $SR_M(\qmf_1) = C$, the common capacity.

Thus $SR_M  (\qmf_1 \times \qmf_2) - SR_M(\qmf_1) - SR_M(\qmf_2)$ is given by
\begin{align*} &\min_{\la \in [0,1]} \big(C + (1-\la) d^* + C + \la d^* \big) - C - C = d^* > 0.  \qedhere
\end{align*}
\end{proof}

\section{Randomized time-division strategy}
\label{sec:randomziedtimedivision}
Randomized time-division refers to a strategy that generalizes the simple time-division strategy. In time-division, the sender $X$ transmits exclusively to receiver $Y$ for a predetermined $\a$ fraction of the time, and  transmits exclusively to receiver $Z$ for the remaining $(1-\a)$ fraction of the time. In randomized time-division, the sender chooses the $\a$ fraction of the time that it wants to transmit to $Y$ using a codebook, thus conveying some commonly decodable information to the receivers when they decode the proper $(\a, 1-\a)$ division of slots. This strategy can be shown to improve on naive time division for some broadcast channels. For more details, an interested reader reader can refer to  \cite[pg. 216]{elk12}.

This is indeed a special (and much simpler) instance of Marton's coding strategy that sets   $U=X, V=\emptyset$ when $W\in \Ac$ and $V=X, U=\emptyset$ when  $W\in\Ac^c$. This strategy yields a $\la$-sum-rate given by
\begin{align*} \la\dash SR_{RTD}(\qmf) & = \max_{p(w,x)} \la I(W;Y) + (1-\la) I(W;Z) + \sum_{w \in \Ac} \P(W=w) I(X;Y|W=w) \\
& \quad \quad  \qquad + \sum_{w \in \Ac^c}\P(W=w) I(X;Z|W=w).\end{align*}
Using standard arguments it follows that it suffices to consider $|\Wc| \leq |\Xc|$ to compute the $\la$-sum-rate.

It was shown \cite{nwg10} that for all binary input broadcast channels the sum rate obtained using the simple randomized time division strategy matches the sum rate obtained using Marton's coding strategy, i.e. $SR_M(\qmf) = SR_{RTD}(\qmf)$ when $|\Xc|=2$. This result is based on the inequality that whenever $|\Xc|=2$ and $(U,V) \to X \to (Y,Z)$ is Markov we have
$$ I(U;Y) + I(V;Z) - I(U;V) \leq \max\{I(X;Y), I(X;Z)\}. $$

Using this inequality it also immediately follows that $\la \dash SR_{RTD}(\qmf) = \la \dash SR_M(\qmf)$.

For  the product of two channels $\qmf_1 \times \qmf_2$ one can define a slight generalization of the RTD strategy (equivalently this is a natural generalization of RTD for the 2-letter channel $\qmf \times \qmf$). This is again a special instance of Marton's coding strategy that sets
$$ (U,V) := \begin{cases} \begin{array}{lc} U=(X_1 X_2), V=\emptyset & w \in \Ac_1 \\ U=X_1 , V=X_2  & w \in \Ac_2 \\ U=X_2 , V=X_1  & w \in\Ac_3\\  U=\emptyset , V=(X_1, X_2)  & w \in \Ac_4 \end{array} \end{cases}, $$
where $\Ac_1, \Ac_2, \Ac_3, \Ac_4$ denotes a partition of $\Wc$.
Let this scheme be called $2\dash RTD$. We define
\begin{align*}
& \la \dash SR_{2\dash RTD}(\qmf_1 \times \qmf_2) \\
&\quad = \max_{p(w,x_1,x_2)} \la I(W;Y_1, Y_2) + (1-\la) I(W;Z_1, Z_2) + \sum_{w \in \Ac_1} \P(W=w) I(X_1, X_2; Y_1, Y_2|W=w) \\
& \qquad + \sum_{w \in \Ac_2} \P(W=w) \big( I(X_1; Y_1, Y_2|W=w) + I(X_2; Z_1, Z_2|W=w) - I(X_1; X_2|W=w) \big) \\
&\qquad + \sum_{w \in \Ac_3} \P(W=w) \big( I(X_1; Y_1, Y_2|W=w) + I(X_2; Z_1, Z_2|W=w) - I(X_1; X_2|W=w) \big) \\
& \qquad + \sum_{w \in \Ac_4} \P(W=w) I(X_1, X_2; Z_1, Z_2|W=w).
\end{align*}

Similarly define
\begin{align*}
& SR_{2\dash RTD}(\qmf_1 \times \qmf_2) \\
&\quad = \max_{p(w,x_1,x_2)} \min\{ I(W;Y_1, Y_2), I(W;Z_1, Z_2)\} +\sum_{w \in \Ac_1} \P(W=w) I(X_1, X_2; Y_1, Y_2|W=w) \\
& \qquad + \sum_{w \in \Ac_2} \P(W=w) \big( I(X_1; Y_1, Y_2|W=w) + I(X_2; Z_1, Z_2|W=w) - I(X_1; X_2|W=w) \big) \\
&\qquad + \sum_{w \in \Ac_3} \P(W=w) \big( I(X_1; Y_1, Y_2|W=w) + I(X_2; Z_1, Z_2|W=w) - I(X_1; X_2|W=w) \big) \\
& \qquad + \sum_{w \in \Ac_4} \P(W=w) I(X_1, X_2; Z_1, Z_2|W=w).
\end{align*}

In a similar fashion to the proof of Lemma \ref{le:mm} one can show the following lemma.
\begin{lemma}
\label{le:mm1} The following holds:
\begin{equation*} \min_{\la \in [0,1]}  \la \dash SR_{2\dash RTD}(\qmf \times \qmf)  = SR_{2\dash RTD}(\qmf).
\end{equation*}
\end{lemma}
The proof is given in Appendix \ref{sec:Appendix:2}.

\begin{remark}
Suppose there is a binary input channel $\qmf(y,z|x)$ such that $SR_{2\dash RTD}(\qmf \times \qmf) > 2 SR_{RTD}(\qmf)$ then it would immediately imply that
$$ SR_{2M}(\qmf) \geq \frac 12 SR_{2\dash RTD}(\qmf \times \qmf) > SR_{RTD}(\qmf) = SR_M(\qmf), $$
where the last equality follows from the result about binary input broadcast channels. This would have been an easy technique to establish the strict sub-optimality of Marton's coding scheme if it had worked. However the next lemma shows that this cannot happen. Indeed we show that $\la \dash SR_{2\dash RTD}(\qmf_1 \times \qmf_2) = \la \dash SR_{RTD}(\qmf_1) + \la \dash SR_{RTD}(\qmf_2)$ for channels with arbitrary input cardinality. Hence from Lemma \ref{le:mm1} it will immediately follow that $SR_{2\dash RTD}(\qmf \times \qmf) = 2 SR_{RTD}(\qmf).$
\end{remark}

\begin{theorem}
\label{th:rtdf}
The following holds:
$$\la \dash SR_{2\dash RTD}(\qmf_1 \times \qmf_2) = \la \dash SR_{RTD}(\qmf_1) + \la \dash SR_{RTD}(\qmf_2).$$
\end{theorem}
\begin{proof}
By taking the product of the optimizing distributions for $\la \dash SR_{RTD}(\qmf_1), \la \dash SR_{RTD}(\qmf_2)$ one can immediately see that
$$\la \dash SR_{2\dash RTD}(\qmf_1 \times \qmf_2) \geq \la \dash SR_{RTD}(\qmf_1) + \la \dash SR_{RTD}(\qmf_2).$$

Hence it suffices to show that
$$\la \dash SR_{2\dash RTD}(\qmf_1 \times \qmf_2) \leq \la \dash SR_{RTD}(\qmf_1) + \la \dash SR_{RTD}(\qmf_2).$$

Observe that
{\small\begin{align}
&\la I(W;Y_1, Y_2) + (1-\la) I(W;Z_1, Z_2) + \sum_{w \in \Ac_1} \P(W=w) I(X_1, X_2; Y_1, Y_2|W=w)  \nonumber\\
& \qquad + \sum_{w \in \Ac_2} \P(W=w) \big( I(X_1; Y_1, Y_2|W=w) + I(X_2; Z_1, Z_2|W=w) - I(X_1; X_2|W=w) \big)  \nonumber\\
&\qquad + \sum_{w \in \Ac_3} \P(W=w) \big( I(X_1; Y_1, Y_2|W=w) + I(X_2; Z_1, Z_2|W=w) - I(X_1; X_2|W=w) \big) \nonumber\\
& \qquad + \sum_{w \in \Ac_4} \P(W=w) I(X_1, X_2; Z_1, Z_2|W=w) \nonumber \\
& = \la H(Y_1, Y_2) + (1-\la) H(Z_1, Z_2) \nonumber\\
&\qquad + \sum_{w \in \Ac_1} \P(W=w) \big(  I(X_1, X_2; Y_1, Y_2|W=w) -\la H(Y_1,Y_2|W=w) - (1-\la) H(Z_1, Z_2|W=w) \big)  \nonumber\\
& \qquad + \sum_{w \in \Ac_2} \P(W=w) \big( I(X_1; Y_1, Y_2|W=w) + I(X_2; Z_1, Z_2|W=w) - I(X_1; X_2|W=w)  \nonumber\\
& \qquad \qquad \qquad -\la H(Y_1,Y_2|W=w) - (1-\la) H(Z_1, Z_2|W=w) \big) \label{eq:2e1}\\
&\qquad + \sum_{w \in \Ac_3} \P(W=w) \big( I(X_1; Y_1, Y_2|W=w) + I(X_2; Z_1, Z_2|W=w) - I(X_1; X_2|W=w)  \nonumber\\
& \qquad \qquad \qquad -\la H(Y_1,Y_2|W=w) - (1-\la) H(Z_1, Z_2|W=w) \big) \nonumber\\
& \qquad + \sum_{w \in \Ac_4} \P(W=w) \big(  I(X_1, X_2; Z_1, Z_2|W=w)  -\la H(Y_1,Y_2|W=w) - (1-\la) H(Z_1, Z_2|W=w) \big). \nonumber
\end{align}}
The idea of the proof is to factorize each of the four summation terms in \eqref{eq:2e1} separately.

Consider the following manipulations of the terms.
\begin{align}
& I(X_1, X_2; Y_1, Y_2|W=w) -\la H(Y_1,Y_2|W=w) - (1-\la) H(Z_1, Z_2|W=w) \nonumber \\
& \quad = I(X_1; Y_1|W=w, Y_2) + I(X_2; Y_2|W=w, Z_1) -\la H(Y_1|W=w, Y_2) \label{eq:2e1s1} \\
& \quad \quad -\la H(Y_2|W=w, Z_1) - (1-\la) H(Z_1|W=w, Y_2) - (1-\la) H(Z_2|W=w, Z_1), \nonumber
\end{align}

\begin{align}
& I(X_1; Y_1, Y_2|W=w) + I(X_2; Z_1, Z_2|W=w) - I(X_1; X_2|W=w)  \nonumber\\
&  \qquad -\la H(Y_1,Y_2|W=w) - (1-\la) H(Z_1, Z_2|W=w)  \nonumber\\
& \quad = I(X_1; Y_1|W=w, Y_2) + I(X_2; Z_2|W=w, Z_1) -\la H(Y_1|W=w, Y_2) \nonumber \\
& \quad \quad -\la H(Y_2|W=w, Z_1) - (1-\la) H(Z_1|W=w, Y_2) - (1-\la) H(Z_2|W=w, Z_1) \nonumber \\
& \quad \quad + I(X_2; Z_1|W=w) + I(X_1; Y_2|W=w, Z_1)- I(X_1; X_2|W=w)  \nonumber\\
& \quad \leq I(X_1; Y_1|W=w, Y_2) + I(X_2; Z_2|W=w, Z_1) -\la H(Y_1|W=w, Y_2) \label{eq:2e1s2}  \\
& \quad \quad -\la H(Y_2|W=w, Z_1) - (1-\la) H(Z_1|W=w, Y_2) - (1-\la) H(Z_2|W=w, Z_1), \nonumber
\end{align}
where the last inequality follows since
$ I(X_1; X_2|W=w)  = I(Z_1, X_1; X_2|W=w) = I(Z_1; X_2|W=w) + I(X_1; X_2|W=w, Z_1) $ $ \geq I(Z_1; X_2|W=w) + I(X_1; Y_2|W=w, Z_1).$
Here we use the fact that $(W,X_2) \to X_1 \to Z_1$ is Markov and $(X_1, Z_1, W) \to X_2 \to Y_2$ is Markov.

In a similar fashion we have
\begin{align}
& I(X_2; Y_1, Y_2|W=w) + I(X_1; Z_1, Z_2|W=w) - I(X_1; X_2|W=w)  \nonumber\\
& \quad \leq I(X_1; Z_1|W=w, Z_2) + I(X_2; Y_2|W=w, Y_1) -\la H(Y_1|W=w, Z_2) \label{eq:2e1s3}  \\
& \quad \quad -\la H(Y_2|W=w, Y_1) - (1-\la) H(Z_1|W=w, Z_2) - (1-\la) H(Z_2|W=w, Y_1) \nonumber
\end{align}

Finally
\begin{align}
& I(X_1, X_2; Z_1, Z_2|W=w) -\la H(Y_1,Y_2|W=w) - (1-\la) H(Z_1, Z_2|W=w) \nonumber \\
& \quad = I(X_1; Z_1|W=w, Z_2) + I(X_2; Z_2|W=w, Y_1) -\la H(Y_1|W=w, Z_2) \label{eq:2e1s4} \\
& \quad \quad -\la H(Y_2|W=w, Y_1) - (1-\la) H(Z_1|W=w, Z_2) - (1-\la) H(Z_2|W=w, Y_1) \nonumber
\end{align}

Define new random variables $W_1, W_2$ having alphabets given by
$$ \Wc_1 = \begin{cases} \begin{array}{lc} (w,z_2) & w \in \Ac_1 \cup \Ac_2, z_2 \in \Zc \\
(w,y_2) & w \in \Ac_3 \cup \Ac_4, y_2 \in \Yc \end{array} \end{cases} \mbox{and} \qquad  \Wc_2 = \begin{cases} \begin{array}{lc} (w,y_1) & w \in \Ac_1 \cup \Ac_2, y_1 \in \Yc \\
(w,z_1) & w \in \Ac_3 \cup \Ac_4, z_1 \in \Zc \end{array} \end{cases}.$$

Further partition $\Wc_1$ into two sets $\Bc$ and $\Bc^c$ according to
$\Bc = \{(w,z_2):  w \in \Ac_1 \cup \Ac_2, z_2 \in \Zc\}$, and partition $\Wc_2$ into two sets $\Cc$ and $\Cc^c$ according to $\Cc = \{(w,y_1):  w \in \Ac_1 , y_1 \in \Yc\} \cup \{(w,z_1):  w \in \Ac_3 , z_1 \in \Zc\}.$

Using \eqref{eq:2e1s1}, \eqref{eq:2e1s2}, \eqref{eq:2e1s3}, \eqref{eq:2e1s4}, and the definitions of $W_1, W_2, \Bc, \Cc$ we can bound the expression in \eqref{eq:2e1} by
\begin{align*}
& \la I(W_1; Y_1) + (1-\la) I(W_1; Z_1)  + \sum_{w_1 \in \Bc} \P(W_1=w_1) I(X_1;Y_1|W_1=w_1) \\
&   \qquad + \sum_{w_1 \in \Bc^c}\P(W_1=w_1) I(X_1;Z_1|W_1=w_1) + \la I(W_2; Y_2) + (1-\la) I(W_2; Z_2) \\
&  \qquad  + \sum_{w_2 \in \Cc} \P(W_2=w_2) I(X_2;Y_2|W_2=w_2)  + \sum_{w_2 \in \Cc^c}\P(W_2=w_2) I(X_2;Z_2|W_2=w_2) \\
& \quad \leq \la \dash SR_{RTD}(\qmf_1) + \la \dash SR_{RTD}(\qmf_2).
\end{align*}

This implies that $$\la \dash SR_{2\dash RTD}(\qmf_1 \times \qmf_2) \leq \la \dash SR_{RTD}(\qmf_1) + \la \dash SR_{RTD}(\qmf_2),$$
and completes the proof of the Lemma.
\end{proof}

\begin{remark}
We wish to bring following unique feature to this proof to the attention of the readers: in identifying the auxiliaries $W_1, W_2$ in terms of $W$, past  or future of $Z$, past or future of $Y$, we actually chose different terms depending on  $w \in \Wc$. This is a freedom that has never been  exploited before (to the best of the knowledge of the authors). A consistent choice does not seem to work here.
\end{remark}

\section{Conclusion}
In this paper we show a variety of results related to Marton\rq{}s inner bound and its optimality. We also show that the tightest known outer bound is strictly sub-optimal. An outer bound is presented for product broadcast channels which is then shown to coincide with Marton\rq{}s inner bound for classes of channels whose capacity regions were previously unknown. This outer bound turns out to be a strict improvement over the previously known tightest outer bound for product broadcast channels. It would be very interesting to extend this outer bound to non-product channels in a natural way. Further a variety of other interesting results are also established which aid in the computation of Marton\rq{}s inner bound.

\bibliographystyle{hieeetr}
\bibliography{mybiblio}

\begin{IEEEbiographynophoto}{Yanlin Geng (M'12)}
Yanlin Geng received his B.Sc. (mathematics) and M.Eng. (signal and information processing) from Peking University, and Ph.D. (information engineering) from The Chinese University of Hong Kong in 2006, 2009, and 2012, respectively. He is currently a postdoctoral researcher in the Information Engineering department at The Chinese University of Hong Kong.
\end{IEEEbiographynophoto}

\begin{IEEEbiographynophoto}{Amin Aminzadeh Gohari (S'10,M'11)}
Amin Aminzadeh Gohari is an Assistant Professor at Sharif University of Technology, Tehran, Iran. Dr. Gohari received his M.Sc. and Ph.D. degree in electrical engineering in 2010 from the University of California, Berkeley, and his B.Sc. degree in 2004 from Sharif University of Technology, Iran. He received the 2010 Eli Jury Award from UC Berkeley, Department of Electrical Engineering, for “outstanding achievement in the area of communication networks,” and the 2009–2010 Bernard Friedman Memorial Prize in Applied Mathematics from UC Berkeley, Department of Mathematics, for “demonstrated ability to do research in applied mathematics.” He also received the Gold Medal from the 41st International Mathematical Olympiad (IMO 2000) and the First Prize from the 9th International Mathematical Competition for University Students (IMC 2002).
\end{IEEEbiographynophoto}

\begin{IEEEbiographynophoto}{Chandra Nair (M'02)}
Chandra Nair is an Associate Professor in the Information Engineering department of the Chinese University of Hong Kong.
Dr. Nair received his Bachelor of Technology (B.Tech) degree in Electrical Engineering from the Indian Institute of Technology (IIT), Madras in 1999. Concurrently, he also completed a four year nurture program in Mathematics at the Institute of Mathematical Sciences (IMSc) under the auspices of the National Board of Higher Mathematics (NBHM).
He received a Masters (2002) and PhD (2005) in electrical engineering from Stanford University.  Subsequently he was a postdoctoral fellow at the theory group in Microsoft Research (Redmond) for two years. Following this he joined the IE department, CUHK, as an assistant professor in Fall 2007.
His research interests are on fundamental problems in various interdisciplinary pursuits involving information theory, combinatorial optimization, statistical physics, and algorithms. 
\end{IEEEbiographynophoto}

\begin{IEEEbiographynophoto}{Yuanming Yu}
Yuanming Yu received his B.Sc. (computer science and technology) from Tsinghua University, and is currently a Ph.D. candidate (computer science and engineering) at The Chinese University of Hong Kong since 2010. 
\end{IEEEbiographynophoto}

\appendix
\section{A Min-Max Theorem}
\label{Apendix0}
\begin{theorem} [Theorem 3 of \cite{ter72}] Let $X$ be a compact connected space, let $Y$ be a set, and let $f: X \times Y \mapsto \mathbb{R}$ be a function satisfying:
\begin{itemize}
\item[$(i)$] For any $y_1, y_2 \in Y$ there exists $y_0 \in Y$ such that
$$ f(x,y_0) \geq \frac 12 \left( f(x,y_1) + f(x,y_2) \right), \forall x \in X. $$
\item[$(ii)$] Every finite intersection of sets of the form $\{x \in X: f(x,y) \leq \alpha)\}$ with $(y,\alpha) \in Y \times R$ is closed and connected.

Then
$$\sup_{y \in Y} \min_{x \in X} f(x,y) = \min_{x \in X} \sup_{y \in Y} f(x,y). $$
\end{itemize}
\label{th:frodo}
\end{theorem}

We now present a Corollary of the above theorem that can be potentially used in many information theory scenarios.

\begin{corollary}
\label{coro:mm} Let $\Lambda_d$ be the $d$-dimensional simplex, i.e. $\lambda_i\geq0$ and $\sum_{i=1}^d\lambda_i=1$. Let $\Pc$ be a set of probability distributions $p(u)$. Let $T_i(p(u)), i=1,..,d$ be a set of functions such that the set $\Ac$, defined by
\begin{align*}\Ac &=\{(a_1,a_2,...,a_d)\in \mathbb{R}^d: a_i\leq T_i(p(u))\mbox{ for some }~p(u) \in \Pc\},
\end{align*}
is a convex set.

Then
$$\sup_{p(u) \in \Pc} \min_{\la \in \Lambda_d} \sum_{i=1}^d \la_i T_i(p(u))  = \min_{\la \in \Lambda_d}\sup_{p(u)\in \Pc}  \sum_{i=1}^d \la_i T_i(p(u)). $$
\end{corollary}

\begin{proof}
Let $f(\la,p(u)) = \sum_{i=1}^d \la_i T_i(p(u)).$
It suffices to verify that $f(\la,p(u))$ satisfies the conditions of Theorem \ref{th:frodo}. Since the set $\Ac$ is convex, we know that for any $p_1(u), p_2(u) \in \Pc$ we have a distribution $p_c(u) \in \Pc$ such that
$$ T_i(p_c(u)) \geq \frac{1}{2} \big( T_i(p_1(u))  + T_i(p_2(u)) \big), i=1,...,d . $$
Hence (using linearity in $\la$ and non-negativity of $\la_i$) we have
$$ f(\la,p_c(u)) \geq \frac{1}{2} \big( f(\la,p_1(u)) + f(\la,p_2(u))\big), \forall \la \in \Lambda_d. $$

Since $f(\la,p(u))$ is a linear function of $\lambda$, it is immediate that the set
$$\Bc(p(u),\alpha) = \{ \la \in \Lambda_d: f(\la, p(u)) \leq \alpha \} $$
is closed for every pair $(p(u), \alpha) \in \Pc \times \mathbb{R}$. Further, due to the linearity in $\la$, if $\la_1, \la_2 \in \Bc(p(u),\alpha)$, then the line segment joining $\la_1$ and $\la_2$ belongs to $\Bc(p(u),\alpha)$. This implies that a finite intersection of sets, each containing $\lambda_1$ and $\lambda_2$ will also contain the line segment joining $\la_1$ and $\la_2$, showing that the finite intersection will be connected. Therefore finite intersections of the sets of the form $\Bc(p(u),\alpha)$ are closed and connected. Thus the Corollary \ref{coro:mm} follows from Theorem \ref{th:frodo}.
\end{proof}

We will now show how one can use the Corollary \ref{coro:mm} to establish Lemma \ref{le:mm}.

\begin{proof} (Proof of Lemma \ref{le:mm})
It is clear that
\begin{align*}
\max_{p(u,v,w,x)} \min_{\la \in [0,1]} \la\dash SR_M(\qmf,p(u,v,w,x)) &  \leq \max_{p(x)} \min_{{ \la \in [0,1]}} \max_{p(u,v,w|x)} \la\dash SR_M(\qmf,p(u,v,w,x)) \\
&  \leq \min_{{\la \in [0,1]}} \max_{p(u,v,w,x)}  \la\dash SR_M(\qmf,p(u,v,w,x)).
\end{align*}
Therefore suffices to show that
$$\max_{p(u,v,w,x)} \min_{\la \in [0,1]} \la\dash SR_M(\qmf,p(u,v,w,x)) =  \min_{{\la \in [0,1]}} \max_{p(u,v,w,x)} \la\dash SR_M(\qmf,p(u,v,w,x)).$$
Here we take $d=2$ and set
\begin{align*}
T_1(p(u,v,w,x)) &= I(W;Y) + I(U;Y|W) + I(V;Z|W) - I(U;V|W) \\
T_2(p(u,v,w,x)) &= I(W;Z) + I(U;Y|W) + I(V;Z|W) - I(U;V|W)
\end{align*}
It is clear that the set
$$ \Ac = \{(a_1, a_2): a_1 \leq T_1(p(u,v,w,x)), a_2 \leq T_2(p(u,v,w,x)) \} $$
is a convex set. (In the standard manner, choose $\Wt = (W,Q)$, such that conditioned on $Q=0$ set $(U,V,W,X) \sim p_1(u,v,w,x)$ and conditioned on $Q=1$ set $(U,V,W,X) \sim p_2(u,v,w,x)$). Hence from Corollary \ref{coro:mm}, we have the proof of Lemma \ref{le:mm}.
\end{proof}

\begin{remark}
The proof of this  lemma in section 3.1.1 of \cite{gea10} is very similar in flavor and uses the convexity of the set $\Ac$. However here we recover it from an application of some general theorems, and this technique and Corollary \ref{coro:mm} may be helpful in other situations as well.
\end{remark}

\section{Computing $\la\dash SR_M$ for the semi-deterministic channel in Fig. \ref{fif:f1}}
\subsection{Maximum of $\la \dash SR_M$ is obtained at the uniform input distribution}
\label{Apendix0.5}
Consider the semi-deterministic channel $\qmf_1$ corresponding to the upper component of the product broadcast channel in Figure \ref{fif:f1}. In this appendix we show that for any $\lambda\in[0,1]$, $\la\dash SR_M(\qmf_1,p(x))$ is less than or equal to $\la\dash SR_M(\qmf_1,u(x))$ where $u$ is the uniform distribution on $\mathcal{X}_1$.
From Lemma \ref{le:conla} note that $\la\dash SR_M(\qmf,p(x))$ is concave in $p(x)$. 

Take an arbitrary $p(x) \sim (a,b,c,d).$ Here $a,b,c,d$ denote the probabilities  assigned (in order) to variables from  top to botton in the upper half of Figure \ref{fif:f1}.
 Because of the symmetry in the component channels  $\qmf_1(y_1|x_1), \qmf_1(z_1|x_1)$ in Figure \ref{fif:f1}, we have
\begin{align*}\la\dash SR_M(\qmf_1,p(x)\sim(a, b, c, d))&= \la\dash SR_M(\qmf_1,p(x)\sim(b, a, d, c))\\&=
\la\dash SR_M(\qmf_1,p(x)\sim(c, d, a, b))\\&=
\la\dash SR_M(\qmf_1,p(x)\sim(d, c, b, a)).\end{align*}
Here we have used the symmetry between inputs $1$ and $2$, and the symmetry between inputs $3$ and $4$, and the symmetry between the pair of inputs $(1,2)$ and $(3,4)$. Using the concavity of $F$, we have
\begin{align*}4\la\dash SR_M(\qmf,p(x)\sim(a, b, c, d))&=\la\dash SR_M(\qmf,p(x)\sim(a, b, c, d))+\la\dash SR_M(\qmf,p(x)\sim(b, a, d, c))+\\&
\la\dash SR_M(\qmf,p(x)\sim(c, d, a, b))+
\la\dash SR_M(\qmf,p(x)\sim(d, c, b, a))\\&\leq 4\la\dash SR_M(\qmf,p(x)\sim(\frac{1}{4}(a+b+c+d), \frac{1}{4}(a+b+c+d),\\&~~~~~~~~~~~~~~~~~~~~~~~~~~~ \frac{1}{4}(a+b+c+d), \frac{1}{4}(a+b+c+d)))
\\&=4\la\dash SR_M(\qmf,u(x)).\end{align*}
\subsection{Computing the $\la$-sum-rate at the uniform input distribution}
\label{Apendix1}
In this appendix we compute $\la\dash SR_M(\qmf_1,u(x))$ at the uniform input distribution for the semi-deterministic  $\qmf_1$ corresponding to the upper component of the product broadcast channel given in Figure \ref{fif:f1}.
\begin{claim}
The $\lambda\mapsto \la\dash SR_M(\qmf_1,u(x))$ curve for the channel under consideration consists of two lines,
$$\la\dash SR_M(\qmf_1,u(x)) = \begin{cases} \begin{array}{ll} \frac 53 - \frac 23 \la & \la \in [0,\frac 12] \\ \frac 43 & \la \in [\frac 12, 1] \end{array} \end{cases}. $$
\end{claim}

\begin{proof} Note that
{\small\begin{align*}\la\dash SR_M(\qmf_1,u(x))&=
\max_{p(u,v,w|x)}\big\{\lambda I(W;Y)+(1-\lambda)I(W;Z)+I(U;Y|W)+I(V;Z|W)-I(U;V|W)\big\}\\&=
\max_{p(u,w|x)}\big\{\lambda I(W;Y)+(1-\lambda)I(W;Z)+I(U;Y|W)+H(Z|UW)\big\}.
\end{align*}}
In the last step we have used the inequality $I(V;Z|W)-I(U;V|W)\leq H(Z|UW)$ together with the fact that $I(Z;Z|W)-I(U;Z|W)= H(Z|UW)$ (thus setting $V=Z$, permissible under the semi-deterministic channel setting, is an optimal choice for $V$). Therefore $\la\dash SR_M(\qmf_1,u(x))$ can be written as
\begin{align*}&\max_{p(u,w|x)}\big\{\la H(Y) + (1-\la) H(Z) + (1-\la) \big( H(Y|W) - H(Z|W) \big) + H(Z|UW) - H(Y|UW)\big\},\end{align*}
which is equal to
\begin{align}&\la H(Y) + (1-\la) H(Z)+\max_{p(w|x)}\big\{(1-\la) \big( H(Y|W) - H(Z|W) \big) +\nonumber\\& \max_{p(u|w,x)}\big(H(Z|UW) - H(Y|UW)\big)\big\}.\label{eqn:AppndB1}\end{align}

Let $\P(X|W=i) = (a_i, b_i, c_i, d_i)$, and $f(a_i, b_i,c_i,d_i) = \max_{p(u|x)} H(Z|U) - H(Y|U)$ conditioned on $\p(X)= (a_i, b_i, c_i, d_i)$. Observe that $f$ is concave. The argument is similar to the one given in Lemma \ref{le:conla} and we will not repeat it here. Further, observe that $f(a_i, b_i,c_i,d_i) = f(b_i, a_i, d_i, c_i)$ because the symmetry between inputs $1$ and $2$, and the symmetry between inputs $3$ and $4$.

Consider the transformation $(a_i, b_i, c_i, d_i) \to (b_i, a_i, d_i, c_i),$ for all $i$ while leaving $\P(W=i)$ unchanged. This preserves expression in equation (\ref{eqn:AppndB1}) because of the symmetry between inputs $1$ and $2$, and the symmetry between inputs $3$ and $4$. Thus the transformation $(a_i, b_i, c_i, d_i) \to (\frac{a_i+b_i}{2},\frac{a_i+b_i}{2}, \frac{c_i+d_i}{2}, \frac{c_i+d_i}{2}),$ for all $i$ while leaving $\P(W=i)$ unchanged, does not decrease the $\la$-sum-rate since $H(Y|W)$ and $f$ are concave functions in $(a_i, b_i,c_i,d_i)$, and $H(Z|W)$ that appears with a negative sign remains constant under this transformation. Therefore without loss of generality assume that $\p(X|W=i) = (\frac{x_i}{2}, \frac{x_i}{2}, \frac{1-x_i}{2}, \frac{1-x_i}{2})$ when optimizing the expression  in equation (\ref{eqn:AppndB1}). Let $\P(W=i) = w_i$. Then we require $\sum w_i x_i = \frac 12$.

Hence we can work out $\la\dash SR_M(\qmf_1,u(x))$ as the maximum over $w_i,x_i$ of the expression 
$$ \la \log 6 + (1-\la) + (1-\la) \sum_i w_i [ \log 3 + \frac 23 - \frac 23 H(x_i,1-x_i)] + \sum_i  w_i f(\frac{x_i}{2}, \frac{x_i}{2}, \frac{1-x_i}{2}, \frac{1-x_i}{2}),$$
over $(w_i,x_i)$ that satisfy $\sum w_i x_i = \frac 12$.

We now compute $f(\frac{x_i}{2}, \frac{x_i}{2}, \frac{1-x_i}{2}, \frac{1-x_i}{2})$. Observe that
\begin{align*}
H(Z) - H(Y) &= H(a +b, c+d) - H(\frac{a + b}{3}, \frac{a + c}{3}, \frac{a + d}{3}, \frac{b + c}{3}, \frac{b + d}{3}, \frac{c + d}{3}) \\
& \stackrel{(a)}{\leq} H(a +b, c+d) - H(\frac{a + b}{3}, \frac{a + c + d}{3}, \frac{a }{3}, \frac{b + c + d}{3}, \frac{b }{3}, \frac{c + d}{3}) \\
& \stackrel{(b)}{\leq} H(a +b, c+d) - H(\frac{a + b}{3}, \frac{a + b + c + d}{3}, \frac{a + b }{3}, \frac{ c + d}{3}, \frac{0 }{3}, \frac{c + d}{3}) \\
& = \frac 13 H(a +b, c+d) - \log 3.
\end{align*}
The step $(a)$ holds because the expression is convex in $c$ and $d$ once we fix $c+d$, therefore its maximum must occur at the boundaries. The step $(b)$ holds because the expression is convex in $a$ and $b$ once we fix $a+b$, therefore its maximum must occur at the boundaries.

Therefore $H(Z) - H(Y)\leq \frac 13 H(a +b, c+d) - \log 3$ for all permissible $(a,b,c,d)$. Since the function $\frac 13 H(a +b, c+d) - \log 3$ is concave, we conclude that $f(a_i, b_i,c_i,d_i) \leq \frac 13 H(a +b, c+d) - \log 3$ for all permissible $(a,b,c,d)$. Hence, at $(a,b,c,d)=(\frac{x_i}{2}, \frac{x_i}{2}, \frac{1-x_i}{2}, \frac{1-x_i}{2})$, we have
$$f(\frac{x_i}{2}, \frac{x_i}{2}, \frac{1-x_i}{2}, \frac{1-x_i}{2}) \leq  \frac 13 H(x_i,1-x_i) - \log 3.$$
The equality can be indeed achieved by taking with probability half $(0, x_i, 0, 1-x_i)$ and with probability half $(x_i, 0, 1-x_i, 0)$. Thus,
$f(\frac{x_i}{2}, \frac{x_i}{2}, \frac{1-x_i}{2}, \frac{1-x_i}{2}) = \frac 13 H(x_i,1-x_i) - \log 3$.

Substituting this in we get
$$ 1 + (1 - \la) \frac 23 + (\frac 13 - \frac 23 (1-\la)) \sum_i w_i H(x_i, 1-x_i).$$
We need to maximize this subject to $\sum w_i x_i = \frac 12$. Clearly when $(1-\la) \leq \frac 12$ the optimal choice is to set $x_i = \frac 12$. This yields a value of $\frac 43$ when $\lambda \geq \frac 12$. In the other interval, it is optimal to set $x_i=0$ w.p. $\frac 12$ and $x_i = 1$ w.p. $\frac 12$. In this case, i.e. $\la \in [0,\frac 12]$, we get $1 + (1-\la)\frac 23 = \frac 53 - \frac 23 \la$.
\end{proof}

\section{Proof of outer bound (Claim \ref{cl:obp}) for product broadcast channels}
\label{sec:pfob}

\begin{proof} Take a code of length $n$. Let $Q$ be a random variable independent of the code book such that $Q$ is uniform in $[1:n]$. Identify \begin{align*}&W_{1} = (M_0, Z_2^{1:n}, Y_{1}^{1:Q-1}, Z_{1}^{Q+1:n}, Q),\\&W_{2} = (M_0, Y_1^{1:n}, Y_{2}^{1:Q-1}, Z_{2}^{Q+1:n}, Q),\\&U_{1} = U_{2} = M_1,\\&V_{1} = V_{2} = M_2,\\&X_1=X_{1Q},\\&X_2=X_{2Q}.\end{align*} We need to verify that these choice of auxiliaries work. We begin with the sum rate. Using the Fano inequality and some manipulations we can write
{\small \begin{align*}
& n (R_0 + R_1 + R_2) - n f_1(\e_n) \\
& \leq \la I(M_0;Y_1^{1:n}, Y_2^{1:n}) + (1-\la) I(M_0;Z_1^{1:n}, Z_2^{1:n}) + I(M_1;Y_1^{1:n}, Y_2^{1:n}|M_0) + I(M_2;Z_1^{1:n}, Z_2^{1:n}|M_0) - I(M_1;M_2|M_0) \\
& = \la I(M_0;Y_1^{1:n}, Y_2^{1:n}) + (1-\la) I(M_0;Z_1^{1:n}, Z_2^{1:n}) + I(M_1;Y_1^{1:n}, Y_2^{1:n}|M_0) + I(M_2;Z_1^{1:n}, Z_2^{1:n}|M_1,M_0) \\&\quad - I(M_1;M_2|M_0, Z_1^{1:n}, Z_2^{1:n}) \\
& = \la I(M_0;Y_2^{1:n}|Y_1^{1:n}) + (1-\la) I(M_0; Z_2^{1:n}) + I(M_1;Y_2^{1:n}|M_0,Y_1^{1:n}) + I(M_2;Z_2^{1:n}|M_1,M_0) \\
& \quad + \la I(M_0;Y_1^{1:n}) + (1-\la) I(M_0; Z_1^{1:n}|Z_2^{1:n})  + I(M_1;Y_1^{1:n}|M_0) + I(M_2;Z_1^{1:n}|M_1,Z_2^{1:n}, M_0) - I(M_1;M_2|M_0, Z_1^{1:n}, Z_2^{1:n}) \\
& =  \la I(M_0;Y_2^{1:n}|Y_1^{1:n}) + (1-\la) I(M_0; Z_2^{1:n}) + I(M_1;Y_2^{1:n}|M_0,Y_1^{1:n}) + I(M_2;Z_2^{1:n}|M_1,M_0) \\
& \quad + \la I(M_0;Y_1^{1:n}) + (1-\la) I(M_0; Z_1^{1:n}|Z_2^{1:n}) + I(M_1;Y_1^{1:n}|M_0) + I(M_2;Z_1^{1:n}|Z_2^{1:n}, M_0) - I(M_1;M_2|M_0,  Z_2^{1:n}) \\
& \leq \la I(M_0;Y_2^{1:n}|Y_1^{1:n}) + (1-\la) I(M_0, Y_1^{1:n}; Z_2^{1:n})+ I(M_1;Y_2^{1:n}|M_0,Y_1^{1:n}) + I(M_2;Z_2^{1:n}|M_1,M_0, Y_1^{1:n}) \\
& \quad + \la I(M_0, Z_2^{1:n};Y_1^{1:n}) + (1-\la) I(M_0; Z_1^{1:n}|Z_2^{1:n}) + I(M_1;Y_1^{1:n}|M_0,Z_2^{1:n}) + I(M_2;Z_1^{1:n}|Z_2^{1:n}, M_0) - I(M_1;M_2|M_0,  Z_2^{1:n}) \\
& \leq \la I(M_0;Y_2^{1:n}|Y_1^{1:n}) + (1-\la) I(M_0, Y_1^{1:n}; Z_2^{1:n})+ I(M_1;Y_2^{1:n}|M_0,Y_1^{1:n}) + I(X_2^{1:n};Z_2^{1:n}|M_1,M_0, Y_1^{1:n}) \\
& \quad + \la I(M_0, Z_2^{1:n};Y_1^{1:n}) + (1-\la) I(M_0; Z_1^{1:n}|Z_2^{1:n}) + I(M_1;Y_1^{1:n}|M_0,Z_2^{1:n}) + I(M_2;Z_1^{1:n}|Z_2^{1:n}, M_0) - I(M_1;M_2|M_0,  Z_2^{1:n}) \\
& \leq \la I(M_0, Y_1^{1:n};Y_2^{1:n}) + (1-\la) I(M_0, Y_1^{1:n}; Z_2^{1:n})+ I(M_1;Y_2^{1:n}|M_0,Y_1^{1:n}) + I(X_2^{1:n};Z_2^{1:n}|M_1,M_0, Y_1^{1:n}) \\
& \quad + \la I(M_0, Z_2^{1:n};Y_1^{1:n}) + (1-\la) I(M_0, Z_2^{1:n}; Z_1^{1:n}) + I(M_1;Y_1^{1:n}|M_0,Z_2^{1:n}) + I(M_2;Z_1^{1:n}|M_0, Z_2^{1:n}) - I(M_1;M_2|M_0,  Z_2^{1:n})
\end{align*}}
where $f_1(\epsilon)$ is a function that converges to zero as $\epsilon$ converges to zero.
Thus,
{\small \begin{align*}
& n (R_0 + R_1 + R_2) - n f_1(\e_n) \\
& \leq \la I(M_0, Y_1^{1:n};Y_2^{1:n}) + (1-\la) I(M_0, Y_1^{1:n}; Z_2^{1:n})+ I(M_1;Y_2^{1:n}|M_0,Y_1^{1:n}) + I(X_2^{1:n};Z_2^{1:n}|M_1,M_0, Y_1^{1:n}) \\
& \quad + \la I(M_0, Z_2^{1:n};Y_1^{1:n}) + (1-\la) I(M_0, Z_2^{1:n}; Z_1^{1:n}) + I(M_1;Y_1^{1:n}|M_0,Z_2^{1:n}) + I(M_2;Z_1^{1:n}|M_0, Z_2^{1:n}) - I(M_1;M_2|M_0,  Z_2^{1:n}).
\end{align*}}
Similarly
{\small \begin{align*}
& n (R_0 + R_1 + R_2) - n f_2(\e_n) \\
& \leq \la I(M_0, Y_1^{1:n};Y_2^{1:n}) + (1-\la) I(M_0, Y_1^{1:n}; Z_2^{1:n})+ I(M_1;Y_2^{1:n}|M_0,Y_1^{1:n}) + I(M_2;Z_2^{1:n}|M_0,Y_1^{1:n})-I(M_1;M_2|M_0,Y_1^{1:n})
\\
& \quad + \la I(M_0, Z_2^{1:n};Y_1^{1:n}) + (1-\la) I(M_0, Z_2^{1:n}; Z_1^{1:n}) + I(M_2;Z_1^{1:n}|M_0, Z_2^{1:n})+I(X_1^{1:n};Y_1^{1:n}|M_0,M_2,Z_2^{1:n}).
\end{align*}}

These lead to the following single letter bounds:
{\small \begin{align*}
R_0 + R_1 + R_2 & \leq \la I(W_2;Y_2) + (1-\la) I(W_2;Z_2) + I(U_2;Y_2|W_2) + I(X_2;Z_2|U_2, W_2) \\
& \quad +  \la I(W_1;Y_1) + (1-\la) I(W_1;Z_1)  + \min \big\{ I(U_1;Y_1|W_1) + I(X_1; Z_1|U_1, W_1), \\
& \qquad \qquad  I(V_1;Z_1|W_1) + I(X_1; Y_1|V_1, W_1) \big\},\\
R_0 + R_1 + R_2 & \leq \la I(W_2;Y_2) + (1-\la) I(W_2;Z_2) + \min \big\{ I(U_2;Y_2|W_2) + I(X_2; Z_2|U_2, W_2), \\
& \qquad \qquad  I(V_2;Z_2|W_2) + I(X_2; Y_2|V_2, W_2) \big\} \\
& \quad +  \la I(W_1;Y_1) + (1-\la) I(W_1;Z_1)  + I(V_1;Z_1|W_1) + I(X_1;Y_1|V_1, W_1).
\end{align*}}

Since the choice of the auxiliaries do not depend on $\lambda$, we conclude that
{\small \begin{align*}
R_0 + R_1 + R_2 & \leq \min \{ I(W_1; Y_1) + I(W_2; Y_2), I(W_1; Z_1) + I(W_2; Z_2) \}  + I(U_2;Y_2|W_2) + I(X_2;Z_2|U_2, W_2) \\
& \qquad + \min \big\{ I(U_1;Y_1|W_1) + I(X_1; Z_1|U_1, W_1), I(V_1;Z_1|W_1) + I(X_1; Y_1|V_1, W_1) \big\},\\
R_0 + R_1 + R_2 & \leq \min \{ I(W_1; Y_1) + I(W_2; Y_2), I(W_1; Z_1) + I(W_2; Z_2) \} + I(V_1;Z_1|W_1) + I(X_1;Y_1|V_1, W_1) \\&\qquad+ \min \big\{ I(U_2;Y_2|W_2) + I(X_2; Z_2|U_2, W_2),   I(V_2;Z_2|W_2) + I(X_2; Y_2|V_2, W_2) \big\}.
\end{align*}}

It remains to verify the following inequalities
{\small \begin{align*}
R_0 & \leq I(W_1; Y_1) + I(W_2; Y_2),\\
R_0 & \leq I(W_1; Z_1) + I(W_2; Z_2), \\
R_0 + R_1 & \leq I(W_1; Y_1) + I(W_2; Y_2)+I(U_1; Y_1|W_1) + I(U_2; Y_2|W_2), \\
R_0 + R_1 & \leq I(W_1; Z_1) + I(W_2; Z_2)+I(U_1; Y_1|W_1) + I(U_2; Y_2|W_2), \\
R_0 + R_2 & \leq I(W_1; Y_1) + I(W_2; Y_2)+I(V_1; Z_1|W_1) + I(V_2; Z_2|W_2),\\
R_0 + R_2 & \leq I(W_1; Z_1) + I(W_2; Z_2)+I(V_1; Z_1|W_1) + I(V_2; Z_2|W_2).
\end{align*}}
The first single-letter formula holds because one can verify that $I(W_1; Y_1)\geq \frac{1}{n}I(M_0;Y_1^{1:n})$ and $I(W_2; Y_2)\geq \frac{1}{n}I(M_0;Y_2^{1:n}|Y_1^{1:n})$. These imply that $I(W_1; Y_1) + I(W_2; Y_2) \geq I(M_0;Y_2^{1:n},Y_1^{1:n})$. One can finish the proof using the Fano inequality. The second inequality on $R_0$ can be proved similarly. The third inequality holds because
$I(U_1W_1; Y_1)\geq \frac{1}{n}I(M_0M_1;Y_1^{1:n})$, $I(U_2W_2; Y_2)\geq \frac{1}{n}I(M_0M_1;Y_2^{1:n}|Y_1^{1:n})$ and $(M_0, M_1)$ can be recovered from $(Y_1^{1:n},Y_2^{1:n})$ with high probability. The fourth inequality holds because:
{\small \begin{align*}
&n(R_0+R_1)-nf_3(\epsilon) \\
& \leq  I(M_0;Z_1^{1:n},Z_2^{1:n})+I(M_1;Y_1^{1:n},Y_2^{1:n}|M_0) \\
& \leq  I(M_0;Z_1^{1:n},Z_2^{1:n})+I(M_1,Z_2^{1:n};Y_1^{1:n}|M_0)+I(M_1;Y_2^{1:n}|M_0, Y_1^{1:n}) \\
& =  I(M_0;Z_1^{1:n}|Z_2^{1:n})+I(M_0;Z_2^{1:n})+I(Y_1^{1:n};Z_2^{1:n}|M_0)+I(M_1;Y_1^{1:n}|M_0, Z_2^{1:n})+I(M_1;Y_2^{1:n}|M_0, Y_1^{1:n}) \\
& \leq  I(M_0, Z_2^{1:n};Z_1^{1:n})+I(M_0, Y_1^{1:n};Z_2^{1:n})+I(M_1;Y_1^{1:n}|M_0, Z_2^{1:n})+I(M_1;Y_2^{1:n}|M_0, Y_1^{1:n}) \\
& =  \sum_{i=1}^n\big(I(M_0, Z_2^{1:n};Z_1^i|Z_{1}^{i+1:n})+I(M_0, Y_1^{1:n};Z_2^i|Z_{2}^{i+1:n})
 +I(M_1;Y_1^i|M_0, Z_2^{1:n}, Y_{1}^{1:i-1}) +I(M_1;Y_2^i|M_0, Y_1^{1:n}, Y_{2}^{1:i-1})\big) \\
 & \leq  \sum_{i=1}^n\big(I(M_0, Z_2^{1:n}, Z_{1}^{i+1:n};Z_1^i)+I(M_0, Y_1^{1:n}, Z_{2}^{i+1:n};Z_2^i)
 +I(M_1;Y_1^i|M_0, Z_2^{1:n}, Y_{1}^{1:i-1}) +I(M_1;Y_2^i|M_0, Y_1^{1:n}, Y_{2}^{1:i-1})\big)\\
 & = \sum_{i=1}^n\big(I(M_0, Z_2^{1:n}, Y_{1}^{1:i-1}, Z_{1}^{i+1:n};Z_1^i)+I(M_0, Y_1^{1:n}, Y_{2}^{1:i-1}, Z_{2}^{i+1:n};Z_2^i)
 +I(M_1;Y_1^i|M_0, Z_2^{1:n}, Y_{1}^{1:i-1})\\
 & \qquad -I(Y_{1}^{1:i-1};Z_1^i|M_0, Z_2^{1:n}, Z_{1}^{i+1:n})
 +I(M_1;Y_2^i|M_0, Y_1^{1:n}, Y_{2}^{1:i-1})-I(Y_{2}^{1:i-1};Z_2^i|M_0, Y_1^{1:n}, Z_{2}^{i+1:n})\big)
 \end{align*}
 
 \begin{align*}
 &=  \sum_{i=1}^n\big(I(M_0, Z_2^{1:n}, Y_{1}^{1:i-1}, Z_{1}^{i+1:n};Z_1^i)+I(M_0, Y_1^{1:n}, Y_{2}^{1:i-1}, Z_{2}^{i+1:n};Z_2^i)
 +I(M_1;Y_1^i|M_0, Z_2^{1:n}, Y_{1}^{1:i-1}) \\
 &\qquad -I(Z_{1}^{i+1:n};Y_1^i|M_0, Z_2^{1:n}, Y_{1}^{1:i-1})
 +I(M_1;Y_2^i|M_0, Y_1^{1:n}, Y_{2}^{1:i-1})-I(Z_{2}^{i+1:n};Y_2^i|M_0, Y_1^{1:n}, Y_{2}^{1:i-1})\big)\\
 & =  \sum_{i=1}^n\big(I(M_0, Z_2^{1:n}, Y_{1}^{1:i-1}, Z_{1}^{i+1:n};Z_1^i)+I(M_0, Y_1^{1:n}, Y_{2}^{1:i-1}, Z_{2}^{i+1:n};Z_2^i)
 +I(M_1;Y_1^i|M_0, Z_2^{1:n}, Y_{1}^{1:i-1}, Z_{1}^{i+1:n})\\
 & \qquad +I(M_1;Y_2^i|M_0, Y_1^{1:n}, Y_{2}^{1:i-1}, Z_{2}^{i+1:n})\big) \\
 & = n(I(W_1; Z_1) + I(W_2; Z_2)+I(U_1; Y_1|W_1) + I(U_2; Y_2|W_2)).
\end{align*}}
The fifth inequality follows in a similar fashion, and the sixth one is similar to the third one. Hence the outer bound is valid.
\end{proof}

\section{Proof of Lemma \ref{le:mm1}}
\label{sec:Appendix:2}
{\em Proof of Lemma \ref{le:mm1}}: This is a consequence of Corollary \ref{coro:mm}. Let $d=2$, let
\begin{align*}
& T_1(p(w,x_1, x_2)) \\
& \quad = I(W;Y_1, Y_2) + \sum_{w \in \Ac_1} \P(W=w) I(X_1, X_2; Y_1, Y_2|W=w) \\
& \qquad + \sum_{w \in \Ac_2} \P(W=w) \big( I(X_1; Y_1, Y_2|W=w) + I(X_2; Z_1, Z_2|W=w) - I(X_1; X_2|W=w) \big) \\
&\qquad + \sum_{w \in \Ac_3} \P(W=w) \big( I(X_1; Y_1, Y_2|W=w) + I(X_2; Z_1, Z_2|W=w) - I(X_1; X_2|W=w) \big) \\
& \qquad + \sum_{w \in \Ac_4} \P(W=w) I(X_1, X_2; Z_1, Z_2|W=w).
\end{align*}

\begin{align*}
& T_2(p(w,x_1, x_2)) \\
& \quad = I(W;Z_1, Z_2) + \sum_{w \in \Ac_1} \P(W=w) I(X_1, X_2; Y_1, Y_2|W=w) \\
& \qquad + \sum_{w \in \Ac_2} \P(W=w) \big( I(X_1; Y_1, Y_2|W=w) + I(X_2; Z_1, Z_2|W=w) - I(X_1; X_2|W=w) \big) \\
&\qquad + \sum_{w \in \Ac_3} \P(W=w) \big( I(X_1; Y_1, Y_2|W=w) + I(X_2; Z_1, Z_2|W=w) - I(X_1; X_2|W=w) \big) \\
& \qquad + \sum_{w \in \Ac_4} \P(W=w) I(X_1, X_2; Z_1, Z_2|W=w).
\end{align*}

It is clear that the set
$$ \Gc = \{(g_1, g_2): g_1 \leq T_1(p(w,x_1, x_2)), g_2 \leq T_2(p(w,x_1, x_2)) \} $$
is a convex set. (In the standard manner, choose $\Wt = (W,Q)$; When $Q=0$ choose $(W,X_1, X_2) \sim p_1(w,x_1, x_2)$ and when $Q=1$ choose $(W,X_1, X_2) \sim p_2(w,x_1, x_2)$). Hence from Corollary \ref{coro:mm}, we have the proof of Lemma \ref{le:mm1}.

\end{document}